\documentclass[conference,letterpaper]{IEEEtran}

\usepackage[utf8]{inputenc} 
\usepackage[T1]{fontenc}
\usepackage{url}
\usepackage{ifthen}
\usepackage{cite}
\usepackage[cmex10]{amsmath} 
\usepackage{tikz}
\usetikzlibrary{spy,backgrounds}
\usepackage{pgfplotstable}

\usepackage{pgfplots}

\usetikzlibrary{arrows.meta,decorations.pathmorphing,positioning,backgrounds,fit,calc,shapes.geometric}
\usepackage{pgfplots}
\usetikzlibrary{shadows} 
\usetikzlibrary{patterns}
\pgfplotsset{compat=newest}

\usepackage{amsmath, amssymb,scalefnt,setspace}
\usepackage{mathrsfs}
\usepackage{amsfonts}
\usepackage{accents}
\usepackage{mathtools}

\usepackage{comment}

\usepackage{enumitem}

\newlist{steps}{enumerate}{1}
\setlist[steps, 1]{wide=0pt, leftmargin=\parindent, label=Step \arabic*:, font=\bfseries}

\usepackage[IEEEtran,notheorems,hyperref]{research19}

\usepackage{xargs}
\usepackage{url}
\usepackage[caption=false,font=footnotesize]{subfig}
\usepackage{stfloats}

\definecolor{myGreen}{rgb}{0,0.7,0.2}

 \newtheorem{theorem}{Theorem}

    \newtheorem{definition}{Definition}
   
   \newtheorem{proposition}{Proposition}
    \newtheorem{example}{Example}
    \newtheorem{conditions}{Conditions}

\newcommandx{\strong}[3][2=\epsilon,3=n]{\mathcal{A}^{* (#3)}_{#2} (#1)}
\newcommandx{\weak}[3][2=\epsilon,3=n]{\mathcal{A}^{ (#3)}_{#2} (#1)}
\newcommandx{\typen}[2][2=n]{\top ^{ (#2)}_{#1}}
\newcommandx{\contypen}[3][3=n]{\top ^{ (#3)}_{#1}(#2)}

\newcommandx{\alltypen}[2][2=n]{\set P_{#2}({#1})}
\newcommandx{\allcontypen}[2][2=n]{\set P_{#2}({#1})}
\newcommandx{\typeseqn}[2][2=n]{\top ^{ (#2)}({#1})}
\newcommandx{\allprob}[1]{\set P({#1})}

\let\vec\relax
\newcommand{\vec}{\mathbf}
\let\set\relax
\newcommand{\set}{\mathcal}

\newcommand{\alphabetX}{\set X}
\newcommand{\alphabetY}{\set Y}

\newcommand{\alphabetN}{\set N}
\newcommand{\alphabetZ}{\set Z}

\newcommand{\alphabetU}{\set U}
\newcommand{\alphabetV}{\set V}
\newcommand{\alphabetT}{\set T}

\newcommand{\tgamma}{\tilde{\gamma}}
\newcommand{\sng}{\sqrt{n}\gamma}

\newcommand{\bfyhat}{{\mathbf{\hat  y}}}
\newcommandx{\repbfyhat}[1][1=P]{\bfyhat_{#1}}
\newcommandx{\optrepbfyhat}[1][1=P]{\bfyhat^*_{#1}}

\newcommand{\pdone}[1]{\frac{\partial}{\partial{#1}}}
\newcommand{\pdtwo}[1]{\frac{\partial^2}{\partial{#1^2}}}
\newcommand{\pdcross}[2]{\frac{\partial^2}{\partial{#1}\partial{#2}}}

\newcommand{\olsi}[1]{\,\overline{\!{#1}}}
  
 \newcommand*{\QEDB}{\null\nobreak\hfill\ensuremath{\square}}%

\newcommand{\ben}{\begin{enumerate}}
\newcommand{\een}{\end{enumerate}}
\newcommand{\bi}{\begin{itemize}}
\newcommand{\ei}{\end{itemize}}

\newcommand{\limn}{\lim_{n\rightarrow\infty}}

\newcommand{\one}{\frac{1}{n}}
\newcommand{\sone}{\frac{1}{\sqrt{n}}}
\newcommand{\half}{\frac{1}{2}}

\newcommand{\onei}{{\rm 1\!\!\!\:I}}

\newcommand{\dist}{\mathsf D}

\newcommand{\timesn}[1]{ {#1}^{\times\!n}}
\newcommand{\qzeron}{Q_0^{\times\!n}}

\newcommandx{\admchannel}[1][1=\dist]{\set W^{\leq #1} }
\newcommandx{\admchanneln}[1][1=n]{\set W^{\leq \dist}_{#1} }
\newcommandx{\discontypen}[2][2=\dist]{\set W^{\leq #2}_n ({#1})}

\newcommandx{\admchannelpermn}[1][1=n]{\bar {\set W}^{\leq \dist}_{#1} }

\newcommand{\changed}[1]{{\color{blue}#1}}

\begin{document}
\title{Covert Capacity of Degraded Broadcast Channels} 


\author{%
  \IEEEauthorblockN{Yossef Steinberg}
  \IEEEauthorblockA{ 
Technion--Israel Institute of Technology\\
Haifa, Israel\\
ysteinbe@technion.ac.il}
  \and
  \IEEEauthorblockN{Mich\`ele Wigger}
  \IEEEauthorblockA{Universit\'e Paris-Saclay, CNRS, CentraleSup\'elec, L2S\\ 
                    91190 Gif-sur-Yvette, France\\
                   michele.wigger@centralesupelec.fr}
}

\maketitle

\begin{abstract}
We derive the capacity region of the degraded broadcast channel (DBC) 
subject to the constraint that the communication is not detected by an adversary, the Warden. Our capacity result is in a computable form and numerical results show that time-sharing is  suboptimal in general, and improved rates can be obtained through superposition coding. 
\end{abstract}


\section{Introduction}\label{sec:intro}
Communication subject to information-theoretic security constraints has a long history in information theory; see, e.g., the seminal work by Wyner \cite{Wyner:75p}. More recently, significant attention has been given to communication systems that are subject to a covertness constraint, i.e., systems in which an external adversary, a so-called Warden, is not allowed to learn even the mere fact that communication is taking place. This requirement is typically enforced by imposing that the Kullback–Leibler divergence (or other measures such as the variational distance) between the warden's actual channel output distribution and the hypothetical output distribution assuming that the transmitter always sends a specific zero symbol $x_0$ remains below a given threshold $\delta > 0$.

The work \cite{bash_first} first showed that, under the above covertness assumption, reliable communication over memoryless Gaussian channels is possible, provided that the number of communicated information bits scales proportionally to the square root of the number of channel uses. Covert rates are therefore commonly defined as $L:= \frac{1}{\sqrt{n\delta}} \log_2 \nu$,  for $n$ the blocklength, $\delta$ the covertness constraint, and $\nu$ the size of the message set. 
Covert capacities for discrete memoryless channels (DMCs) and Gaussian memoryless channels were determined in \cite{bash_first, bash_p2p, Bloch:16p, WangWornellZheng:16p}, and were shown not to depend on the parameter $\delta$. Moreover, \cite{Bloch:16p} also determined the rate of the secret key shared between encoders and decoders required to achieve this common capacity. This result was extended to all covert rates (not only capacity) in \cite{key}.

Covert capacity regions for discrete memoryless multi-access interference channels  were determined in  \cite{bloch_k_users_mac} and \cite{Tin}, and the required key rates at all covert rates in \cite{key}. 
Broadcast channels (BC) under a covertness constraint were studied in  \cite{bloch_journal_embedding_broadcast, ligong_broadcast, TanLee:19p}. The works in   \cite{bloch_journal_embedding_broadcast, ligong_broadcast} considered a mixed covert/non-covert scenario  with a non-covert communication from the transmitter to both receivers and a covert communication from  the transmitter to only one of the receivers, which needs to remain undetectable (covert) to the other receiver. In contrast, the work in \cite{TanLee:19p} and this present article both  consider a scenario where the transmitter sends individual messages to two broadcast receivers, and the entire communication needs to remain undetectable to an external warden. 

The time-sharing region in this setup can be written as \cite{TanLee:19p}: 
\begin{equation}\label{eq:TS}
\mathcal{L}^{(\textnormal{TS})}=\left \{(L_1,L_2)\colon 
\frac{L_1}{L_1^*}+  \frac{L_2}{L_2^*} \leq 1\right\}
\end{equation} for $L_1^*$ and $L_2^*$ the   covert capacities  of the BC marginal channel transition laws $P_{Y_1|X}$ and $P_{Y_2|X}$. 
 The work in \cite{TanLee:19p} proved that time-sharing is optimal for all   BCs with  marginal capacities satisfying  $L_1^* \geq L_2^*$and \begin{equation}
\label{eq:Time_sharing_c1}
\frac{L_1^*}{L_2^*}\geq \sup_{P_X}\frac{I(X;Y_1)}{I(X;Y_2)}.
\end{equation}
The same work also proved optimality of time-sharing for all Gaussian BCs and binary symmetric BCs. 

In this work, we show that optimality of time-sharing does not hold for general (stochastically or physically) degraded channels, and superposition coding can strictly improve over the time-sharing region. In fact, we present a computable characterization of the superposition coding region and show that it achieves capacity for all degraded BCs.

\section{Definitions}
\label{sec:def}
Let $\alphabetX$, $\alphabetY_1$, $\alphabetY_2$, $\alphabetZ$ be finite sets. 
Denote by ${\cal P}(\alphabetX)$ the class of all distributions on $\alphabetX$.
A discrete memoryless broadcast channel (BC) with two users and a warden is a quintuple
$\{\alphabetX,P_{Y_1,Y_2,Z|X},\alphabetY_1,\alphabetY_2,\alphabetZ\}$
where $\alphabetX$ is the input alphabet, $\alphabetY_k$ is the output alphabet of user $k$, for $k=1,2$, 
$\alphabetZ$ is the output alphabet at the warden and $P_{Y_1,Y_2,Z|X}$ is a transition probability matrix
from $\alphabetX$ to $\alphabetY_1\times\alphabetY_2\times\alphabetZ$. 
We denote by $P_{Z|X}$ and $P_{Y_k|X}$, for $k=1,2$, the conditional marginals of $P_{Y_1,Y_2,Z|X}$,
and by $P_X\circ P_{Z|X}$ the distribution on $\alphabetZ$ induced by $P_X$ at the input. 
When $P_X$ is understood from the context, we will use the notation $P_Z$, and similarly
for $P_X\circ P_{Y_k|X}$, $P_{Y_k}$ etc.
Let $x_0\in\alphabetX$ stand for the zero symbol, 
namely, the symbol fed into the channel when no communication is taking place.
Define
\begin{equation}
Q_0(z) = P_{Z|X}(z|x_0),\quad z\in\alphabetZ,
\label{eq:def_Qz}
\end{equation}
and let $\qzeron$ stand for  the $n$-fold product of $Q_0$, i.e.,
\begin{equation}
\qzeron(z^n) = \prod_{i=1}^n Q_0(z_i). \label{eq:qzeron}
\end{equation}
The goal in covert communication is to transmit information to the receivers
while keeping the distribution at the warden output close to $\qzeron$.
Fix integers $\nu_k$, $k=1,2$ and transmission length $n$. 
Let $\alphabetN_k=[1:\nu_k]$ stand for the set of messages of user $k$.
The transmitter and legitimate receivers share a random secret key $S$, taking values in a finite set ${\set S}$.
We assume that the key is of sufficiently large randomness, thus do not specify the size of ${\set S}$.
  \begin{definition}
 \label{def:code}
 An $(n,\nu_1,\nu_2,\epsilon,\delta)$ covert code for the BC
 $\{\alphabetX,P_{Y_1,Y_2,Z|X},\alphabetY_1,\alphabetY_2,\alphabetZ\}$
 with a warden consists of an encoder
 \begin{equation}
 f: \alphabetN_1\times\alphabetN_2\times{\set S} \rightarrow \alphabetX^n, \label{eq:def_enc}
 \end{equation}
 and pair of decoders
 \begin{equation}
 \phi_k: \alphabetY^n_k\times {\set S}\rightarrow \alphabetN_k,\qquad k=1,2,\label{eq:def_dec}
 \end{equation}
 such that the probabilities of error are bounded by $\epsilon$:
 \begin{IEEEeqnarray}{rCl}
 \lefteqn{
 P_{e,k} =} \nonumber\\
 &&\frac{1}{\nu_1\nu_2}\!\sum_{\substack{(m_1,m_2) \\\in \alphabetN_1\times\alphabetN_2}}\sum_{s\in{\set S}} P_S(s)
  P_{Y_k^n|X^n}\left(D_{m_k}^c(s) | f(m_1,m_2,s)\right)\leq\epsilon,\nonumber\\
 && \hspace{6.5cm}\quad k=1,2,
 \label{eq:def_pe}
 \end{IEEEeqnarray}
 and the output distribution at the warden approximates $\qzeron$ in the divergence sense:
 \begin{equation}
 D\left( P_{Z^n}|| \qzeron\right) \leq \delta.
 \label{eq:def_divergence}
 \end{equation}
 The set $D_{m_k}(s)$ in~(\ref{eq:def_pe}) is the decoding region of message $m_k$
 \begin{equation}
 D_{m_k}(s) = \left\{y_k^n:\quad \phi_k(y_k^n,s) = m_k\right\}, \qquad k=1,2,
 \label{def:dec_region}
 \end{equation}
 and $P_{Z^n}$ in~(\ref{eq:def_divergence}) stands for the distribution of the warden output $Z^n$ induced by  the operation of  the code:
 \begin{IEEEeqnarray}{rCl}
P_{Z^n}(z^n) =
 &&  \frac{1}{\nu_1\nu_2}\sum_{\substack{(m_1,m_2) \\\in \alphabetN_1\times\alphabetN_2}}\sum_{s\in{\set S}} P_S(s)
  P_{Z^n|X^n}(z^n|f(m_1,m_2,s)). \nonumber\\\label{eq:def_pzn}
 \end{IEEEeqnarray}
 \end{definition}

In this work, we focus on BCs  that satisfy the following conditions,
that are now standard in covert communications \cite{WangWornellZheng:16p,Bloch:16p,TanLee:19p}.
\begin{conditions}[Non-redundancy and absolute continuity]
\label{cond:Q0}
\begin{enumerate}[label=\alph*)]
\item The zero symbol is not redundant at the warden output. I.e., $Q_0\not\in\mbox{CH}[P_{Z|X}(\cdot|x'),\ x'\in\alphabetX\setminus\{x_0\}]$,     
         where CH stands for the convex hull.
\item Absolute continuity w.r.t. $x_0$ symbol at the warden: $P_{Z|X}(\cdot|x)\ll Q_0\quad \forall x\in\alphabetX$.
\item Absolute continuity w.r.t. $x_0$ symbol at the users output: 
         $P_{Y_k|X}(\cdot|x)\ll P_{Y_k|X}(\cdot|x_0)\quad \forall x\in\alphabetX,\ \ k=1,2$.
\end{enumerate}
\end{conditions}
Part a) of Conditions~\ref{cond:Q0} guarantees that the encoder cannot mimic the no communication state
with some input distribution $P_X$, that results with output distribution at the warden being equal to $Q_0$. 
If part b)  is not satisfied, than there is an input symbol $x'$ that the encoder cannot use, effectively 
reducing the input alphabet size. For the single user channel, it is shown in~\cite{Bloch:16p} that if
c) does not hold, the number of covert bits that can be transmitted
grows like $\sqrt{n}\log n$. Thus, Conditions~\ref{cond:Q0} are the most pessimistic assumptions
that still allow covert communications. For details, see~\cite[Appendix G]{Bloch:16p} and~\cite{TanLee:19p}.

The covert rates of the code are defined as
\begin{equation}
L_k=\frac{\log \nu_k}{\sqrt{n\delta}},\quad k=1,2
\label{eq:covert_rates}
\end{equation}
A pair of covert rates $(L_1,L_2)$ is called $\delta$-achievable if for any $\epsilon>0$, $\rho>0$ and sufficiently large $n$
there exists an $(n,2^{\sqrt{n\delta}(L_1-\rho)},2^{\sqrt{n\delta}(L_2-\rho)},\epsilon,\delta)$ covert code for the BC $P_{Y_1,Y_2,Z|X}$.
The collection of all  $\delta$-achievable pairs is called the covert capacity region, and is denoted by ${\set L}_\delta^*$. As we will see, it does not depend on the value of $\delta>0$. 

In this work we derive the covert capacity  region for stochastically degraded broadcast channels, where the degradation is between
the legitimate users, i.e., we assume  that there exists a conditional distribution $P_{Y_2|Y_1}$ such  that
\begin{equation}
P_{Y_2|X}(y_2|x) = \sum_{y_1} P_{Y_1|X}(y_1|x)P_{Y_2|Y_1}(y_2|y_1).\label{eq:deg}
\end{equation}
By the problem definition,
${\set L}_\delta^*$ depends on~$P_{Y_1,Y_2,Z|X}$ only via its conditional marginals. 
Hence in the sequel a BC with a warden is referred to as 
three channels with common input $\{P_{Y_1|X},P_{Y_1|X},P_{Z|X}\}$, where the alphabets are understood from the context. 
In addition, no distinction has to be made between
stochastically and physically degraded models, and they are  commonly referred to as degraded channels.

\section{Main results}
\label{sec:main_results}
Let ${\set L}_{n,\delta}^{(I)}$ stand for the collection of  nonnegative pairs $(L_1,L_2)$ satisfying
\begin{subequations}
\label{eq:cap_limit}
\begin{IEEEeqnarray}{rCl}
L_1 &\leq& \sqrt{n/\delta} I(X;Y_1|U)\label{eq:L1_limit}\\
L_2 &\leq& \sqrt{n/\delta} I(U;Y_2)\label{eq:L2_limit}
\end{IEEEeqnarray}
where mutual informations are calculated according to $P_{U,X} P_{Y_1Y_2|X}$ for some $P_{U,X}$ so that the induced $P_Z= P_X \circ P_{Z|X}$ satisfies 
\begin{equation}
D(P_Z|| Q_0) \leq \frac{\delta}{n}. \label{eq:div_constraint_limit}
\end{equation}
Define
\begin{equation}
{\set L}_\delta^{(I)} = \bigcap_{n\geq 1} {\set L}_{n,\delta}^{(I)}. \label{eq:cap_intersection}
\end{equation}
\end{subequations}
Before proceeding to our  main result, we state a few properties
of  the region~${\set L}_\delta^{(I)}$. 
To exhaust ${\set L}_\delta^{(I)}$  it is enough to restrict the alphabet $\alphabetU$ to satisfy
\begin{equation}
\left|\alphabetU\right|\leq \left|\alphabetX\right|+1.
\label{eq:U_size}
\end{equation}
The bound~(\ref{eq:U_size}) is proved using the support lemma~\cite{CsiszarKorner:82b}
for every $n$.
Note that the presence of the additional constraint~(\ref{eq:div_constraint_limit}) does not increase the alphabet
size of $U$, because when applying the support lemma to restrict $|\alphabetU|$,
the distribution of $X$ is preserved, hence also~(\ref{eq:div_constraint_limit}).
The details are omitted.

Let $\check{P}_{U,X}^{(n)}$ be a sequence
of distributions that achieves a rate pair $(l_1,l_2)\in{\set L}_\delta^{(I)}$. Since the alphabets
are finite, ${\set P}(\alphabetU\times\alphabetX)$ is compact, hence $\check{P}_{U,X}^{(n)}$
converges to a limit distribution on a subsequence $n_k,\ k=1,2...$, with $n_k<n_{k+1}$. Define a sequence
of distributions $P^{(n)}_{U,X}$ as follows:
\begin{subequations}
\begin{IEEEeqnarray}{rCl}
P^{(n_k)}_{U,X} &=& \check{P}_{U,X}^{(n_k)}\quad k=1,2...\label{subeqeq:def_P_converge_1}\\
P^{(n)}_{U,X}    &=& \check{P}_{U,X}^{(n_k)} \quad n_{k-1}<n\leq n_k. \label{subeqeq:def_P_converge_2}
\end{IEEEeqnarray}
\end{subequations}
Then $P_{U,X}^{(n)}$ converges, and achieves the same rate pair $(l_1,l_2) \in {\set L}_\delta^{(I)}$.
To simplify notation, from this point on we drop the superscript $(n)$ and use just $P_{U,X}$ with the understanding
that the distributions depend on $n$ and converge.
    
\begin{theorem}
\label{theo:main1}
For any discrete memoryless degraded BC with a warden, 
the following holds:
\begin{enumerate}
\item
${\set L}_\delta^* ={\set L}_\delta^{(I)}$.
\item
A sequence of distributions $P_{U,X}$ achieves a  positive $L_2$ according to constraint \eqref{eq:L2_limit} 
only if there exists a subset $B\subset\alphabetU$ such that
\begin{subequations}
\label{subeq:decay_U_theo}
\begin{IEEEeqnarray}{rCl}
P_U(B) &>& 0\label{subeq:decay_U_pos_theo}\\
\limn P_U(B) &=& 0. \label{subeq_decay_U_lim_theo}
\end{IEEEeqnarray}
\end{subequations}
\end{enumerate}
\end{theorem}


 See Appendix~\ref{sec:main1_proof}.

Note that although the size of $\alphabetU$ is finite and fixed,  
${\set L}_\delta^{(I)}$
is still not a computable result, since it involves the limit $n\rightarrow\infty$.  The following 
computable region $\tilde{\set L}^{(I)}$ coincides with ${\set L}_\delta^{(I)}$, which is stated in Theorem~\ref{cor:L_region_2} and proved in 
 Section~\ref{sec:main_computable} ahead. 

Let $\tilde{\set L}^{(I)}$ stand for the collection of  nonnegative pairs $(L_1,L_2)$  satisfying: 
\begin{subequations} 
\begin{IEEEeqnarray}{rCl}
L_1 &\leq&\sqrt{\frac{2}{ \chi_{2}(\nu) }}   \bigg[ (1\!-\!\nu)  \sum_x  \!\tilde{P}_X^A(x) D\left(P_{Y_1|X}(\cdot|x)||P_{Y_1|X}(\cdot|x_0) \right) 
              \nonumber \\
              &&\hspace{2cm}+ \nu  I^B(Y_1; X| U) 
           \bigg]   \label{subeq:computable_region_L1C}\\
L_2 &\leq& \sqrt{\frac{2}{ \chi_{2}(\nu) }}  \bigg[   \nu\sum_x {P}_X^B(x) D\left(P_{Y_2|X}(\cdot|x)||P_{Y_2|X}(\cdot|x_0) \right)   \nonumber \\
&&. \hspace{2cm} - \nu  I^B(Y_2;X|U)  \bigg],
               \label{subeq:computable_region_L2C}          
\end{IEEEeqnarray}
\end{subequations}
for some $\nu \in [0,1]$, auxiliary alphabet $\mathcal{U}$ of size not exceeding $|\mathcal{X}|+1$, a  singleton set $A =\{ u_0\} \subset \mathcal{U}$, its complement $B= \mathcal{U}\backslash A$, and pmfs 
$P_{UX}^{B}$ over $B\times \mathcal{X}$ and $\tilde{P}_{X}^{A}$ over $\mathcal{X}\backslash\{x_0\}$, where in the above   mutual informations are  calculated with respect to the  pmf $P_{UX}^B$ and
\begin{IEEEeqnarray}{rCl} 
 \chi_{2}(\nu) & := & \chi_{2} \left( (1- \nu) \tilde{P}_Z^A +\nu {P}_Z^B  \| Q_0 \right),
 \end{IEEEeqnarray} 
 where $\chi_2(\cdot \| \cdot)$ denotes the $\chi_2$-distance: 
 \begin{equation} 
 \chi_2\left( P \| Q\right) :=  \sum_{z \in \mathcal{Z}} \frac{\left(  {P}(z) -Q(z) \right)^2}{Q(z)}. 
 \end{equation} 
%
%

\begin{theorem}
\label{cor:L_region_2} 
It holds that  $\tilde{\set L}^I= {\set L}_\delta^{(I)}$.
\end{theorem}
The proof of above theorem is given in Section~\ref{sec:main_computable}.

The time-sharing region ${\set L}^{(\textnormal{TS})}$  is obviously  included in our region, see  Appendix~\ref{app:Ts}.

\subsection{Comparison to Time-sharing and Numerical Examples}
In \cite{TanLee:19p}, it was shown that when~\eqref{eq:Time_sharing_c1} holds,
time-sharing is optimal and suffices to achieve $\mathcal{L}_{\delta}$. We will reprove this result using our 
capacity-expression in \eqref{subeq:computable_region_L1C} and \eqref{subeq:computable_region_L2C}.

Notice first that  \eqref{eq:Time_sharing_c1} in particular holds when $X$ is  binary with probability $1-\alpha$ equal to $x_0$ and  with probability $\alpha$ equal to $x$, for arbitrary $x\in \mathcal{X}\backslash \{x_0\}$  and $\alpha>0$. Letting $\alpha \to 0$, we can conclude that for any $x \in \mathcal{X}\backslash\{x_0\}$: 
\begin{equation} 
\frac{L_1^*}{L_2^*} \geq \lim_{\alpha \to 0} \frac{I(X;Y_1)}{I(X;Y_2)} = \frac{D(P_{Y_1|X}
(\cdot|x)  \|P_{Y_1|X}(\cdot|x_0))}{  D(P_{Y_2|X}(\cdot|x) \| P_{Y_2|X}(\cdot|x_0))},\label{eq:ineqD}
\end{equation} because for above choice of $X$ and when $\alpha \to 0$ we have $I(X;Y_k)= \alpha D(P_{Y_k|X}(\cdot|x) \| P_{Y_k|X}(\cdot|x_0))\cdot (1+o(1))$, for $k =1,2$.

If in the following expression we apply above inequality \eqref{eq:ineqD} and \eqref{eq:Time_sharing_c1} on the individual summands, we can write
 \begin{IEEEeqnarray}{rCl}
\lefteqn{\nu  \sum_{u \in B}  P_U^B(u)  \frac{ I^B(Y_1; X| U=u) }{L_1^*} } \qquad \nonumber \\[-1ex]
&+&
             (1-\nu) \sum_x \tilde{P}_X^A(x) \frac{ D\left(P_{Y_1|X}(\cdot|x)\|P_{Y_1|X}(\cdot|x_0) \right) }{L_1^*}  \nonumber\\
 &\leq&   \nu  \sum_{u \in B}  P_U^B(u)  \frac{ I^B(Y_2; X| U=u) }{L_2^*} \nonumber \\[-1ex]
 &&           +  (1-\nu) \sum_x \tilde{P}_X^A(x) \frac{ D\left(P_{Y_2|X}(\cdot|x)||P_{Y_2|X}(\cdot|x_0) \right) }{L_2^*}.\IEEEeqnarraynumspace \label{eq:ineq}
\end{IEEEeqnarray}

Plugging \eqref{eq:ineq} into the upper bound on $\frac{L_1}{L_1^*}+\frac{L_2}{L_2^*}$ obtained from \eqref{subeq:computable_region_L1C}--\eqref{subeq:computable_region_L2C}, allows to conclude that for channels satisfying \eqref{eq:Time_sharing_c1},  any achievable pair $(L_1,L_2)$ lies in the time-sharing region $\mathcal{L}^{(\textnormal{TS})}$ defined in \eqref{eq:TS}.  
For details, see Appendix~\ref{app:TS_opt}.

\begin{example} 
Consider a ternary input alphabet $\mathcal{X}=\{0,1,2\}$ and quaternary output alphabet 
$\mathcal{Y}_1=\mathcal{Y}_2=\mathcal{Z}=\{0,1,2,3\}$. Let $x_0=0$ and consider the following channel transition laws 
\begin{equation}  P_1= \left[ \begin{matrix} 0.2& 0.28& 0.28& 0.24\\ 
       0.05& 0.1&  0.45& 0.4\\
      0.07& 0.37& 0.4& 0.16.
\end{matrix}
\right]
\end{equation}
and
\begin{equation}  P_2= \left[ \begin{matrix} 
   0.1884&   0.324  & 0.232  & 0.2556\\
   0.0515 &  0.215  & 0.331   &0.4025\\
   0.0744  & 0.399 &  0.326  & 0.2006
   \end{matrix}
\right]
\end{equation}
for the legitimate receivers and 
\begin{equation} 
  Q=\left[ \begin{matrix}0.20 & 0.19 &  0.36& 0.25\\
     0.01&  0.37 &0.17&  0.45\\
   0.42& 0.35 & 0.05 &  0.18
\end{matrix}
\right]
\end{equation} 
for the warden. 
  Notice that the channel from $Y_1$ to $Y_2$ is stochastically degraded  because we can write 
$P_2=P_1\cdot W$ for 
\begin{equation} 
  W=\left[ \begin{matrix}0.9& 0.1&0& 0\\
  0.02& 0.8& 0.12& 0.06\\
  0.01& 0.2& 0.7& 0.09\\
  0&0.1& 0.01& 0.89\end{matrix}
\right].
\end{equation} 
  The covert capacities for the two single-user  channels $P_1$ and $P_2$ in the presence of the warden 
  $Q$ are $L_1^*=  0.46809$ and $L_2^* =  0.28590$. Figure~\ref{fig:ex1} shows the  boundary of the region 
  $\mathcal{L}_{\delta}$ (solid line) and the boundary of the time-sharing region ${\set L}^{(\textnormal{TS})}$ (dashws line). 
  We observe that for this example, superposition coding improves over time-sharing whenever  $L_1<L_1^*$ or $L_2< L_2^*$. 

\begin{figure}[h!]
\centering
\begin{tikzpicture}
\begin{axis}[%
     width=2.5in,
    height=1.5in,
    at={(0in,0in)},
    scale only axis,
    xmin=0,
    xmax=0.5,
    xlabel style={font=\color{white!15!black}},
    xlabel={Rate $L_1$},
    ymin=0,
    ymax=0.3,
    xtick={0, 0.2, 0.4},
        ytick={0, 0.1, 0.2, 0.3},
    ylabel style={font=\color{white!15!black}},
    ylabel={Rate $L_2$},
    axis background/.style={fill=white},
    title style={font=\bfseries},
    axis x line*=bottom,
    axis y line*=left,
    legend style={
        legend cell align=left, 
        align=left,       
        draw=white!15!black,
        at={(.8,0.85)}, 
        anchor=west}
]
\addplot [color=blue, solid,line width=1.5pt]
  table[row sep=crcr]{
   0.46809   0.00000\\
   0.46830   0.00033\\
   0.46831   0.00034\\
   0.46832   0.00036\\
   0.46832   0.00037\\
   0.46833   0.00039\\
   0.46835   0.00044\\
   0.46836   0.00048\\
   0.46836   0.00050\\
   0.46836   0.00056\\
   0.46836   0.00058\\
   0.46836   0.00062\\
   0.46835   0.00065\\
   0.46834   0.00067\\
   0.46832   0.00072\\
   0.46829   0.00077\\
   0.46821   0.00086\\
   0.46811   0.00095\\
   0.46805   0.00100\\
   0.46268   0.00512\\
   0.46144   0.00606\\
   0.45515   0.01086\\
   0.44906   0.01549\\
   0.44672   0.01728\\
   0.41729   0.03964\\
   0.39896   0.05356\\
   0.39875   0.05371\\
   0.39232   0.05857\\
   0.38543   0.06355\\
   0.37261   0.07277\\
   0.36466   0.07839\\
   0.36210   0.08014\\
   0.35591   0.08433\\
   0.32594   0.10459\\
   0.32073   0.10808\\
   0.31742   0.11026\\
   0.31582   0.11131\\
   0.31425   0.11234\\
   0.24808   0.15522\\
   0.24544   0.15692\\
   0.24217   0.15903\\
   0.23935   0.16084\\
   0.23914   0.16098\\
   0.23246   0.16522\\
   0.23126   0.16598\\
   0.23010   0.16671\\
   0.22898   0.16742\\
   0.22789   0.16810\\
   0.17023   0.20428\\
   0.15591   0.21327\\
   0.14548   0.21982\\
   0.14308   0.22126\\
   0.14066   0.22262\\
   0.13822   0.22391\\
   0.13443   0.22570\\
   0.09593   0.24374\\
   0.04939   0.26472\\
   0.00000   0.28590\\
};
\addlegendentry{$\mathcal{L}^*_{\delta}$}

\addplot [color=red, dashed,line width=1.5pt]
  table[row sep=crcr]{%
   0.00000   0.28590\\
   0.46809   0.00000\\
};
\addlegendentry{$\mathcal{L}^{(\textnormal{TS})}$}

\end{axis}
\end{tikzpicture}%

\caption{Illustration of the capacity region $\mathcal{L}_\delta$ (solid line) and the time-sharing region $\mathcal{L}^{(\textnormal{TS})}$  (dashed line).}
\label{fig:ex1}
\end{figure}
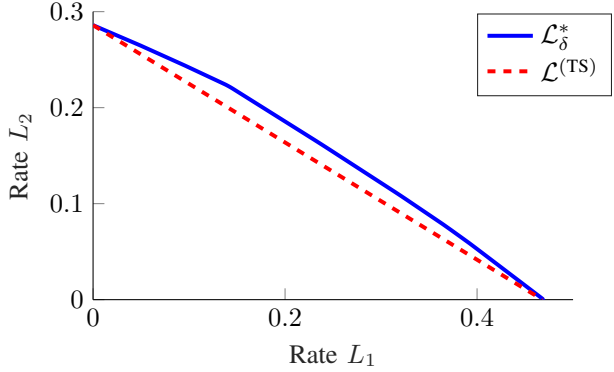
 
\end{example} 

\begin{example} 
 Consider a second example with binary inputs $\mathcal{X}=\{0,1\}$, for $x_0=0$, 
 and ternary outputs $\mathcal{Y}_1=\mathcal{Y}_2=\mathcal{Z}=\{0,1,2\}$. 
 Let the channel to the strong receiver $P_1$ be a BSC$(0.2)$ and the channel to the warden $Q$ be a BSC($0.4$).  
 The channel to the weaker receiver $P_2=P_1 \cdot W$, for
  \begin{equation}  W= \left[ \begin{matrix} 0.9& 0.1\\ 
   c & 1-c
\end{matrix}
\right],
\end{equation}
where we study different values of $c\in\{0, 0.1, 0.2, \ldots, 1\}$. Table \ref{tab:your_label} shows the maximum coefficient 
\begin{equation} 
\gamma^*= \max_{(L_1, L_2)\in \mathcal{L}_\delta} \left(\frac{L_1}{L_1^*}+  \frac{L_2}{L_2^*}\right)
\end{equation} for different values of $c$. This parameter $\gamma^*$ captures by how much one can improve over the time-sharing region, for which the parameter cannot exceed $1$. The second column of Table~\ref{tab:your_label} indicates whether the condition 
$\frac{L_1^*}{L_2^*} \geq \max_{P_X} \frac{I(X;Y_1)}{I(X;Y_2)}$ is satisfied (indicated by 1 in the table) or not (indicated with a 0). 
It has been shown in \cite{TanLee:19p} that time-sharing is optimal and thus $\gamma^*=1$ whenever the  condition is satisfied. 
Our results seem to imply that for this example one can improve over time-sharing whenever the condition does not hold.

\begin{table}[h]
\centering
\begin{tabular}{c|c|c}
\hline
$c$  & $\frac{L_1^*}{L_2^*} \geq \max_{P_X} \frac{I(X;Y_1)}{I(X;Y_2)}$ & $\gamma^*$ \\
\hline
0.0 & 1 & 1.0000 \\
0.1 & 1 & 1.0000 \\
0.2 & 0 & 1.0047 \\
0.3 & 0 & 1.0108 \\
0.4 & 0 & 1.0153 \\
0.5 & 0 & 1.0178 \\
0.6 & 0 & 1.0178 \\
0.7 & 0 & 1.0148 \\
0.8 & 0 & 1.0078 \\
0.9 & 1 & 1.0000 \\
1.0 & 1 & 1.0000 \\
\hline
\end{tabular}
\caption{Results for our binary-input channel example.}
\label{tab:your_label}
\end{table}
\end{example}

\section{Proof of Theorem~\ref{cor:L_region_2}:  Computable Capacity Characterization}
\label{sec:main_computable}
We construct here a sufficiently general joint distribution $P_{U,X}$
that adheres to~(\ref{eq:div_constraint_limit}). A fully-general distribution can be obtained by letting in the following the chosen distributions $P_U^B$, $P_U^A$, $P_{X|U}$, and $\tilde{P}_{X|U}$ depend on $n$. However, a close inspection reveals that the constraints $L_1$ and $L_2$ only depend on limiting points of these distributions and not on how they evolve with $n$. To avoid cumbersome notation,  we therefore assume these probability laws to be constant.

Let $B$ be a proper subset of $\alphabetU$
and $A$ its complement. Let $P_U^B$ (resp. $P_U^A$) be a general distribution  on $B$ (resp. on $A$),
$P_{X|U}(\cdot|u)$  a  conditional distribution on $\alphabetX$ for $u\in B$,
and $\tilde{P}_{X|U}(\cdot |u)$ a conditional distribution on $\alphabetX \setminus x_0$ for $u\in A$.
With these definitions, we set
\begin{subequations}
\label{subeq:structure_UX}
\begin{IEEEeqnarray}{rCl}
\hat{P}_U(u) &=& (1-\mu_1) P_U^A(u) +\mu_1 P_U^B(u) \label{subeq:structure_U}\\
\hat{P}_{X|U}(x|u) &=& \left\{ \,
                                        \begin{IEEEeqnarraybox}[] [c] {l?s}
                                        \IEEEstrut
                                        (1-\mu_2)\onei_{x_0}(x) + \mu_2 \tilde{P}_{X|U}(x|u) & for $u\in A$, \\
                                        P_{X|U}(x|u) & for $u\in B$,
                                        \IEEEstrut
                                        \end{IEEEeqnarraybox}
                                     \right. \nonumber \\ 
                                       \label{subeq:structure_XgU}
\end{IEEEeqnarray}
where 
\begin{equation}
\tilde{P}_{X|U}(x_0|u) = 0\quad \forall u\in A,\label{eq:Pt_XgU}
\end{equation}
\end{subequations}
and $\mu_1$, $\mu_2$ are small parameters that tend to $0$ as $n\rightarrow\infty$, at rates to be determined later. 
The structure we suggest in~(\ref{subeq:structure_U}) and~(\ref{subeq:structure_XgU})
determines $P_{U,X}$, and thus also $P_{X,Z}$ and $P_{U,X,Y_k}$, $k=1,2$.
We define below the notation needed for the characterization of the computable region.
The distribution of $X$ is given by
\begin{subequations}
\label{subeq:structure_UX_2}
\begin{IEEEeqnarray}{rCl}
\hat{P}_X(x)& =& \sum_u\hat{P}_{X|U}(x|u) \hat{P}_U(u) \\
&=& \olsi{\mu}_1\olsi{\mu}_2\onei_{x_0}(x)
  +\olsi{\mu}_1\mu_2\tilde{P}_X^A(x) +\mu_1 P_X^B(x) \label{subeq:hat_px}
\end{IEEEeqnarray}
where $\onei_{x_0}$ puts mass 1 on $x_0$, and we use the notation
\begin{IEEEeqnarray}{rCl}
\tilde{P}_X^A(x) &\eqdef& \sum_{u\in A} \tilde{P}_{X|U}(x|u) P_U^A(u)\label{subeq:tpxA_def}\\
P_X^B(x) &\eqdef& \sum_{u\in B} P_{X|U}(x|u) P_U^B(u)\label{subeq:pxB_def}
\end{IEEEeqnarray}
\end{subequations}
Note that $\tilde{P}_X^A(x_0) = 0$, and
$\hat{P}_X\rightarrow\onei_{x_0}$ as $\mu_1, \mu_2\rightarrow 0$. 
For simplicity of exposition we also define
\begin{subequations}
\label{subeq:structure_UX_3}
\label{subeq:1}
\begin{IEEEeqnarray}{rCl}
\tilde{P}_Z^A(z) &\eqdef& \sum_x  P_{Z|X}(z|x) \tilde{P}_X^A(x)\label{subeq:tilde_P_Z_A}\\
P_Z^B(z) &\eqdef& \sum_x  P_{Z|X}(z|x) P_X^B(x)\label{subeq:P_Z_B}\\
\tilde{P}_{Y_k|U}^A(y_k|u) &\eqdef& \sum_x P_{Y_k|X}(y_k|x) \tilde{P}_{X|U}(x|u)\quad\mbox{for}\ u\in A,\nonumber\\
&& \hspace{4cm}  k=1,2. \label{subeq:tP_YkgU_A}\\
P_{Y_k|U}^B(y_k|u) &\eqdef& \sum_x P_{Y_k|X}(y_k|x) P_{X|U}(x|u)\quad\mbox{for}\ u\in B,\nonumber\\&&
\hspace{4cm} k=1,2. \label{subeq:P_YkgU_B}
\end{IEEEeqnarray}
\end{subequations}
Observe that $P_U^A$, $P_U^B$ and~(\ref{subeq:tpxA_def}--\ref{subeq:P_YkgU_B}) do not depend on $\mu_1,\mu_2$. 
Define now the normalized parameters 
\begin{subequations}
\label{subeq:normalized_parameters}
\begin{IEEEeqnarray}{rCl}
\eta_1 &=&\sqrt{n/\delta}\mu_1\label{subeq:normalized_eta1}\\*
\eta_2 &=& \sqrt{n/\delta}\mu_2.\label{subeq:normalized_eta_2}
\end{IEEEeqnarray}
\end{subequations}

Theorem~\ref{cor:L_region_2} is obtained by evaluating the region $\mathcal{L}^{(I)}$ for above choice of distributions based on the Taylor expansions of the terms $I(U;Y_2)$, $I(X;Y_1|U)$
and $D(P_Z||Q_0)$ near $\mu_1=0$, $\mu_2=0$. 

As proved in Appendix~\ref{sec:computable_proof}, this Taylor expansion results in the rate expressions 
\begin{subequations}
\label{subeq:computable_region}
\begin{IEEEeqnarray}{rCl}
L_1 &\leq& \bigg[ \sum_x P_X^B(x) D\left(P_{Y_1|X}(\cdot|x)||P_{Y_1|X}(\cdot|x_0) \right) \nonumber \\&& 
              \;\; - \sum_{u\in B} P_U^B(u) D\left(P_{Y_1|U}(\cdot|u)||P_{Y_1|X}(\cdot|x_0)\right) \bigg]
              \eta_1\nonumber\\
              & & + \sum_x \tilde{P}_X^A(x) D\left(P_{Y_1|X}(\cdot|x)||P_{Y_1|X}(\cdot|x_0) \right) 
              \eta_2 \label{subeq:computable_region_L1}\\
L_2 &\leq&     \sum_{u\in B} P_U^B(u) D\left(P_{Y_2|U}(\cdot|u)||P_{Y_2|X}(\cdot|x_0)\right)\eta_1,
               \label{subeq:computable_region_L2}          
\end{IEEEeqnarray}
while the divergence constraint evaluates to
\begin{IEEEeqnarray}{rCl}
\lefteqn{\eta_1^2 \chi_2(P_Z^B||Q_0) 
  +\eta_2^2 \chi_2(\tilde{P}_Z^A||Q_0) } \quad \nonumber \\
  && + \eta_1\eta_2\sum_{z \in \mathcal{Z}} \frac{(\tilde{P}_Z^A(z) - Q_0(z))(P_Z^B(z)-Q_0(z))}{Q_0(z)} \leq 2. \IEEEeqnarraynumspace \label{subeq:computable_region_gamma_eta}
\end{IEEEeqnarray} 
\end{subequations}


\medskip
Without loss in optimality, in the parametrization above we can restrict the set ${A}$ to be a singleton 
(=because the result only depends on $\tilde{P}_X^A$)=. Similarly, the rate-constraints are loosest if $\eta_1$ and $\eta_2$ are chosen
so that constraint \eqref{subeq:computable_region_gamma_eta} is satisfied with equality. 
We thus reparametrize  $\eta_1$ and $\eta_2$ as $\eta_1=c \cdot \nu$ and $\eta_2=c \cdot (1-\nu)$ 
for $\nu\in[0,1]$   and $c>0$,  where the latter should be chosen to ensure equality in  \eqref{subeq:computable_region_gamma_eta} 
we obtain the characterization in the theorem.

\section*{Acknowledgment}

{This work was supported by the ERC under Grant Agreement 101125691.}

\begin{appendices}

\section{Proof of Theorem~\ref{theo:main1}}
\label{sec:main1_proof}

\subsection{Proof of Part~2)}

By the single-user results in~\cite{WangWornellZheng:16p}, it is clear that we can achieve
positive rates $(L_1,L_2)$ (apply a simple time-sharing scheme). Therefore, the distributions $P_{U,X}$
that maximize the outer bound in Theorem~\ref{theo:main1} 
under the constraint~(\ref{eq:div_constraint_limit}), should yield $(I(X;Y_1|U),I(U;Y_2))$
that decay like $n^{-1/2}$ as $n\rightarrow\infty$. We claim that this can be achieved only
when the $U$-marginal of $P_{U,X}$ has a set $B\subset\alphabetU$ whose probability decays to $0$
as $n\rightarrow\infty$.

\begin{proposition}
\label{prop:decay_U}
$I(U;Y_2) =O(n^{-1/2})$ only if there exists a subset $B\subset\alphabetU$ such that
\begin{subequations}
\label{subeq:decay_U}
\begin{IEEEeqnarray}{rCl}
P_U(B) &>& 0\label{subeq:decay_U_pos}\\
\limn P_U(B) &=& 0. \label{subeq_decay_U_lim}
\end{IEEEeqnarray}
\end{subequations}
\end{proposition}

\noindent
{\bf Proof:} The requirement~(\ref{eq:div_constraint_limit}) implies 
the following structure on $P_X$ (\cite[eq.~(32)]{WangWornellZheng:16p}):
\begin{equation}
\hat{P}_{X}(x) = (1-\mu)\onei_{x_0}(x) +\mu\tilde{P}_X(x) \label{eq:Phat_X}
\end{equation}
where $\onei_{x_0}(x)$ (resp. $\tilde{P}_X$) puts mass 1 (resp. 0) on $x_0$, and
\begin{equation}
\mu = O(n^{-1/2}).\label{eq:mu_limit}
\end{equation}
Due to  the Markov chain $U\markov X\markov Y_2$, (\ref{eq:Phat_X}) and~(\ref{eq:mu_limit}), we have
\begin{equation}
\limn I(U;Y_2) = 0.\label{eq:I_UY2_lim}
\end{equation}
Therefore we can write
\begin{equation}
I(U;Y_2) = \left. \frac{\partial}{\partial \mu} I(U;Y_2)\right\rvert_{\mu=0} \mu +O(\mu^2), \label{eq:taylor_I}
\end{equation}
where
\begin{IEEEeqnarray}{rCl}
\frac{\partial}{\partial \mu} I(U;Y_2) &=& \frac{\partial}{\partial \mu} \sum_{u,y_2} P_{U,Y_2}(u,y_2)
         \log\frac{P_{U,Y_2}(u,y_2)}{P_U(u)P_{Y_2}(y_2)}\nonumber\\
       &=& \sum_{u,y_2} \left[ \pdone{\mu} P_{U,Y_2}(u,y_2)\right] \log\frac{P_{Y_2|U}(y_2|u)}{P_{Y_2}(y_2)}. \IEEEeqnarraynumspace\label{eq:der_1_I}     
 \end{IEEEeqnarray}
 For the proof of~(\ref{eq:der_1_I}), see Appendix~\ref{appendix:proof_der1_I}. 
 Fix $\alpha>0$, and assume that 
 \begin{equation}
 P_U(u')\geq \alpha\quad \forall n \label{eq:uprime_condition}
\end{equation}
for some $u'\in\alphabetU$. Then we must have
\begin{equation}
\limn P_{X|U}(x_0|u') = 1, \label{eq:uprime_condition2}
\end{equation} 
as otherwise~(\ref{eq:div_constraint_limit}) does not hold.
Thus, if~(\ref{eq:uprime_condition}) holds for all $u\in\alphabetU$,  
then~(\ref{eq:uprime_condition2}) holds for all elements of $\alphabetU$,
resulting in
\begin{IEEEeqnarray}{rCl}
\limn P_{Y_2|U}(y_2|u) &=& \limn \sum_x P_{Y_2|X}(y_2|x)P_{X|U}(x|u) \IEEEeqnarraynumspace\\
& = &P_{Y_2|X}(y_2|x_0).
\label{eq:y2u_1}
 \end{IEEEeqnarray}
Moreover, by~(\ref{eq:Phat_X}), we also have
\begin{equation}
\limn P_{Y_2}(y_2) = P_{Y_2|X}(y_2|x_0). \label{eq:y2u_2}
\end{equation}
By the structure of $\hat{P}_X$
the derivative of $P_{U,Y_2}$
according to $\mu$ is bounded, so~(\ref{eq:der_1_I}), (\ref{eq:y2u_1}), (\ref{eq:y2u_2})
and~(\ref{eq:taylor_I})
yield
\begin{equation}
I(U;Y_2) = O(\mu^2).
\end{equation}
 Therefore, a necessary condition for having $I(U;Y_2)= O(\mu)$ is that some of  the elements
 of $\alphabetU$ have vanishing probabilities as $n\rightarrow\infty$.
\QEDB

\subsection{Converse  for Part 1)}
\label{subsec:main1_converse_proof}
Assume we have a sequence of 
$(n,2^{\sqrt{n\delta}L_1},2^{\sqrt{n\delta}L_2},\epsilon_n,\delta)$ codes with $\limn\epsilon_n=0$.
Denote by $M_k$ the  random message for user $k$, $k=1,2$.
By Fano's inequality
\begin{IEEEeqnarray}{rCl}
\sqrt{n\delta}L_2(1-\epsilon_n)&\stackrel{(a)}{=}&\log\nu_2(1-\epsilon) \\
&\leq & I(M_2;Y_2^n|S) \\
&= &\sum_{i=1}^n I(M_2;Y_{2,i}|S,Y_2^{i-1})\\
&\leq& \sum_{i=1}^n I(M_2Y_2^{i-1};Y_{2,i}|S)\\
&\leq& \sum_{i=1}^n I(M_2Y_2^{i-1}Y_1^{i-1};Y_{2,i}|S) \\
&\stackrel{(b)}{=} &\sum_{i=1}^n I(M_2Y_1^{i-1};Y_{2,i}|S)\\
&\leq& \sum_{i=1}^n I(M_2Y_1^{i-1}S;Y_{2,i}),\label{eq:Fano_M2}
\end{IEEEeqnarray}
where $(a)$ is by~(\ref{eq:covert_rates}) and in $(b)$ we use the Markov chain $X\markov Y_1\markov Y_2$. Similarly,
\begin{IEEEeqnarray}{rCl}
\sqrt{n\delta}L_1(1-\epsilon_n) &\leq& I(M_1;Y_1^n|S,M_2) \\
&=& \sum_{i=1}^n I(M_1;Y_{1,i}|S,M_2,Y_1^{i-1})
                                                                \\
                  &\stackrel{(a)}{=}& \sum_{i=1}^n I(M_1X_i;Y_{1,i}|S,M_2,Y_1^{i-1})\\
                  &\stackrel{(b)}{=}&
                                                 \sum_{i=1}^n I(X_i;Y_{1,i}|S,M_2,Y_1^{i-1})\label{eq:Fano_M1}
\end{IEEEeqnarray}
where $(a)$ holds since $X^n$ is a deterministic function 
of $(M_1,M_2)$ and $(b)$ due to the Markov chain
\begin{equation}
(M_1,M_2,Y_1^{i-1},Y_2^{i-1},S)\markov X_1\markov (Y_{1,i}Y_{2,i}).\label{eq:Markov2}
\end{equation}
Define
\begin{equation}
U_i = (M_2,Y_1^{i-1},S),\label{eq:Ui_def}
\end{equation}
so that after normalization~(\ref{eq:Fano_M2}, \ref{eq:Fano_M1}) read 
\begin{IEEEeqnarray}{rCl}
 L_2(1-\epsilon_n) &\leq& \frac{1}{\sqrt{n\delta}}\sum_{i=1}^n I(U_i;Y_{2,i}) \label{eq:L2}\\
L_1(1-\epsilon_n) &\leq& \frac{1}{\sqrt{n\delta}} \sum_{i=1}^n I(X_i;Y_{1,i}|U_i). \label{eq:L1}
\end{IEEEeqnarray}
The bounds~(\ref{eq:L1_limit}, \ref{eq:L2_limit}) 
follow from~(\ref{eq:L2}, \ref{eq:L1}) by the classical time sharing argument.
Inequality~(\ref{eq:div_constraint_limit}) is proved exactly like~\cite[eq. (9), Theorem 1]{WangWornellZheng:16p}.

\subsection{Direct Part for Part 1)}
\label{subsec:main1_direct_proof}

The proof of the achievability part of Theorem~\ref{theo:main1}, proceeds along the following  steps:
\begin{steps}
  \item Obtain a layered (superposition) version of Feinstein's Lemma~\cite{Feinstein:54p} for the BC, from the results
           of~\cite{LiuCuffVerdu:15c}.

  \item Show that distributions $P^{(n)}_{U,X}$  with $X$ marginal satisfying~(\ref{eq:div_constraint_limit}) 
             stabilize the information spectrum expressions
            of Step 1. I.e., the information spectrum random variables converge, as $n$
            tends to $\infty$, to the mutual information functions of the outer bound.
\end{steps}

Fix a joint distribution $P_{UVTX}$ such that the Markov chain $(U,V,T)\markov X\markov (Y_1,Y_2)$
holds.
We use the following notation for the mutual information random variables~\cite{VerduHan:94p,Han:02b}:
\begin{subequations}
\label{subeq:information_rv_def}
\begin{IEEEeqnarray}{rCl}
i_{X;Y_1}(X;Y_1) &=& \log\frac{P_{Y_1|X}(Y_1|X)}{P_{Y_1}(Y_1)}\label{subeq:information_rv_def1}\\
i_{V;Y_1|U}(V;Y_1|U) &=& \log\frac{P_{Y_1|U}(Y_1|U,V)}{P_{Y_1|U}(Y_1|U)}\label{subeq:information_rv_def2}
\end{IEEEeqnarray}
\end{subequations}
and similarly for $i_{U;Y_2}(U;Y_2)$ etc.

{\bf Step 1.}  For convenience, we state here the one-shot coding result of~\cite{LiuCuffVerdu:15c}
for general BCs.
Note that the alphabets are of arbitrary size, hence there is no dependence on $n$.
We use the notation of~\cite{LiuCuffVerdu:15c}, but do not repeat their definitions,
for space considerations.
\begin{theorem}[Theorem~10 in \cite{LiuCuffVerdu:15c}]
\label{theo:general_BC}
Fix a BC $P_{Y_1,Y_2|X}$, a joint distribution $P_{UVT}$, a map
$f:\alphabetU\times\alphabetV\times\alphabetT\rightarrow\alphabetX$,
and integers $M_0$, $M_{1,0}$, $M_{2,0}$, $N$, $L$, $\hat{N}$ and $\hat{L}$. Set
\begin{subequations}
\label{subeq:one_shot_bc_par}
\begin{IEEEeqnarray}{rCl}
M&=& M_0 M_{1,0} M_{2,0},\label{subeq:one_shot_bc_par0}\\
M_1 &=& M_{1,0} N, \label{subeq:one_shot_bc_par1}\\
M_2 &=& M_{2,0} L, \label{subeq:one_shot_bc_par2}\\
\tilde{N} &=& \hat{N} N ,\label{subeq:one_shot_bc_par3}\\
\tilde{L} &=& \hat{L} L.\label{subeq:one_shot_bc_par4}
\end{IEEEeqnarray}
\end{subequations}
\end{theorem}
Then, for any $\gamma>0$ there exists an $(M_0,M_1,M_2,\epsilon_1,\epsilon_2)$ code for the BC
with
\begin{IEEEeqnarray}{rCl}
\max\{\epsilon_1,\epsilon_2\}
&\leq& 2\exp(-\gamma)+\exp^{-\exp(\gamma)}\nonumber\\
&  & +  P\left[\left\{ i_{UV;Y_1}(UV;Y_1) \leq \log M\tilde{N} + \gamma\right\}\right.\nonumber\\
& & \hspace{1cm} \cup \left\{ i_{UT;Y_2}(UT;Y_2) \leq \log M\tilde{L} + \gamma\right\} \nonumber\\
& &  \hspace{1cm}  \cup \left\{ i_{V;Y_1|U}(V;Y_1|U) \leq \log \tilde{N} + \gamma\right\} \nonumber\\
& &  \hspace{1cm}  \cup \left\{ i_{T;Y_2|U}(T;Y_2|U) \leq \log \tilde{L} + \gamma\right\} \nonumber\\
 & & \hspace{1cm}  \left. \cup \left\{ i_{V;T|U}(V;T|U) > \log \hat{N}\hat{L} - 2 \gamma\right\} \right] \nonumber\\
 & & + \frac{\min\left\{\hat{N},\hat{L}\right\} - 1}{\hat{N}\hat{L}\left(\exp(-\gamma)-\exp(-2\gamma)\right)}\IEEEeqnarraynumspace
 .\label{eq:pe_one_shot}
\end{IEEEeqnarray}
For our use, we choose the following random variables and parameters in Theorem~\ref{theo:general_BC}. 
For $V$ and $T$:
\begin{subequations}
\label{subeq:one_shot_parameters}
\begin{IEEEeqnarray}{rCl}
V & = & X;\label{subeq:one_shot_parameters_V}\\
T & = & \emptyset;\label{subeq:one_shot_parameters_T}
\end{IEEEeqnarray}
and for $M_0$, $M_{1,0}$, $L$, $\hat{N}$ and $\hat{L}$:
\begin{equation}
M_0 = M_{1,0} = L =\hat{N} = \hat{L} = 1.\label{subeq:one_shot_parameters_1}
\end{equation} 
With~(\ref{subeq:one_shot_parameters_1}) we obtain
\begin{IEEEeqnarray}{rCl}
M_1 &=& N;\label{subeq:one_shot_parameters_M1}\\
M_2 & = & M_{2,0};\label{subeq:one_shot_parameters_M_2}\\
M &=& M_2;\label{subeq:one_shot_parameters_M}\\
\tilde{N} &=& M_1;\label{subeq:one_shot_parameters_tildeN}\\
\tilde{L} &=& 1.\label{subeq:one_shot_parameters_tildeL}
\end{IEEEeqnarray}
\end{subequations}
Substituting~(\ref{subeq:one_shot_parameters}) in~(\ref{eq:pe_one_shot}) and using the union bound,
we conclude that there exists a $(1,M_1,M_2,\epsilon,\epsilon)$ code for the BC with
\begin{IEEEeqnarray}{rCl}
\epsilon &\leq& 2\exp(-\gamma)+\exp^{-\exp(\gamma)}\nonumber\\ 
& & + P\left[i_{X;Y_1}(X;Y_1) \leq \log M_1M_2 + \gamma\right] \nonumber\\
              & &  +  P\left[i_{U;Y_2}(U;Y_2) \leq \log M_2 + \gamma\right] \nonumber\\
              & & + P\left[i_{X;Y_1|U}(X;Y_1|U) \leq \log M_1 + \gamma\right] \label{eq:FeinsteinBC}
\end{IEEEeqnarray}
where $P_X$ is induced by $P_{U,V,T}$ and the mapping $f$. Note that here $X$ is not a deterministic
function of $U$, due to~(\ref{subeq:one_shot_parameters_V}). Hence $P_{U,X}$ is a general 
joint distribution.

We now pass to fixed alphabets and transmission length $n$. In~(\ref{eq:FeinsteinBC}),
$\tgamma$ is arbitrary. Thus choose an arbitrary $\tgamma>0$ and set
\begin{equation}
\sqrt{n}\tgamma = \gamma. \label{eq:gamma_tgamma}
\end{equation}
Using~(\ref{eq:FeinsteinBC}) and the notation of Definition~\ref{def:code},
we conclude that for every $P_{U^n,X^n}$ on $\alphabetU^n\times\alphabetX^n$ 
such that $P_{X^n}$ satisfies~(\ref{eq:def_divergence}),
and any $\tgamma>0$, there exists an~$(n,\nu_1,\nu_2,\epsilon,\delta)$ covert code for the BC
with
\begin{IEEEeqnarray}{rCl}
\epsilon &\leq& 2\exp(-\sqrt{n}\tgamma)+\exp^{-\exp(\sqrt{n}\tgamma)}\nonumber\\ 
& & + P\left[\sone i_{X^n;Y_1^n}(X^n;Y_1^n) \leq \sone\log \nu_1 \nu_2 + \tgamma\right] \nonumber\\
              & &    + P\left[\sone i_{U^n;Y_2^n}(U^n;Y_2^n) \leq \sone \log \nu_2 + \tgamma\right] \nonumber\\
              & & + P\left[\sone i_{X^n;Y_1^n|U}(X^n;Y_1^n|U^n) \leq \sone \log \nu_1 + \tgamma\right]. \IEEEeqnarraynumspace \label{eq:FeinsteinBC_n}
\end{IEEEeqnarray}
This completes Step~1.

{\bf Step 2.}
Let $P_{U,X}$ be any joint distribution satisfying the covertness constraint \eqref{eq:div_constraint_limit}, which implies: 
\begin{equation}
\limn P_X(x_0) = 1. \label{eq:limit_px0}
\end{equation}
 Then, let $P_{U^n,X^n}$ be the $n$-fold product of $P_{U,X}$:
\begin{equation}
P_{U^n,X^n}(u^n,x^n) = P_{U,X}^{\times\!n}(u^n,x^n) = \prod_{i=1}^n P_{U,X}(u_i,x_i).
\label{eq:pux_n-prod}
\end{equation}

We show next that the random variables in~(\ref{eq:FeinsteinBC_n}) converge in probability
to the corresponding single letter information functions, i.e.,
\begin{subequations}
\label{subeq:i_convergence}
\begin{IEEEeqnarray}{rCl}
 \sone i_{X^n;Y_1^n}(X^n;Y_1^n) &\longrightarrow& \sqrt{n}I(X;Y_1) \quad \mbox{in prob.}
                \label{subeq:i_convergence_XY1}\\
 \sone i_{U^n;Y_2^n}(U^n;Y_2^n) &\longrightarrow&   \sqrt{n}I(U;Y_2) \quad  \mbox{in prob.}\IEEEeqnarraynumspace
                \label{subeq:i_convergence_UY2}\\
 \sone i_{X^n;Y_1^n|U^n}(X^n;Y_1^n|U^n) &\longrightarrow& \sqrt{n}I(X;Y_1|U) \quad \mbox{in prob.} \nonumber \\
                \label{subeq:i_convergence_XY1gU}
\end{IEEEeqnarray}
\end{subequations}
The proof of~(\ref{subeq:i_convergence_XY1}) follows exactly the lines 
of the proof of~\cite[eq.~(16)]{WangWornellZheng:16p} using~(\ref{eq:limit_px0}) and is omitted.
The proof of~(\ref{subeq:i_convergence_UY2}) follows these lines as well, using~(\ref{subeq:decay_U}).
We give it here for completeness. First, note that
\begin{IEEEeqnarray}{rCl}
\mathsf{E} \sone i_{U^n;Y_2^n}(U^n;Y_2^n) &=& 
\mathsf{E}\sone \log\frac{\timesn{P}_{Y_2|U}(Y_2^n|U^n)}{P_{Y_2}^{\times\!n}(Y_2^n)}\\
&=& \sqrt{n} I(U;Y_2). \label{eq:i_UY2_conv1}
\end{IEEEeqnarray}
Hence by Chebyshev's inequality 
\begin{IEEEeqnarray}{rCl}
\IEEEeqnarraymulticol{3}{l}{
P\left[ \left| \sone i_{U^n;Y_2^n}(U^n;Y_2^n) - \sqrt{n} I(U;Y_2)\right| \geq \alpha\right]}\nonumber\\* \quad
  &\leq& \frac{1}{\alpha^2} \mathsf{var}\left(\sone \log \frac{\timesn{P}_{Y_2|U}(Y_2^n|U^n)}{P_{Y_2}^{\times\!n}(Y_2^n)}\right),
  \label{eq:i_UY2_conv2}
\end{IEEEeqnarray}
so to prove~(\ref{subeq:i_convergence_UY2}) it is enough to show that
\begin{equation}
\limn \mathsf{var}\left(\sone \log\frac{\timesn{P}_{Y_2|U}(Y^n|U^n)}{P_Y^{\times\!n}(Y^n)}\right) = 0.
 \label{eq:i_UY2_conv3}
\end{equation}
Indeed
\begin{IEEEeqnarray}{rCl}
\lefteqn{\mathsf{var}\left(\sone \log\frac{\timesn{P}_{Y_2|U}(Y^n|U^n)}{P_Y^{\times\!n}(Y^n)}\right) } \quad 
\nonumber\\   &=& \one\sum_{i=1}^n \mathsf{var}\left( \log\frac{\timesn{P}_{Y_2|U}(Y_{2,i}|U_i)}{P_{Y_2}(Y_{2,i})}\right)
 \\
    &=& \mathsf{var}\left(\log\frac{P_{Y_2|U}(Y_2|U)}{P_{Y_2}(Y_2)} \right) \\
    &\leq& \mathsf{E}\left[\left( \log\frac{P_{Y_2|U}(Y_2|U)}{P_{Y_2}(Y_2)} \right)^{\!2}\ \!\right]\\
    &=& \sum_{u\in A}P_U(u)\sum_{y_2} P_{Y_2|U}(y_2|u)\left( \log\frac{P_{Y_2|U}(y_2|u)}{P_{Y_2}(y_2)} \right)^{\!2}
    \nonumber\\
    & & +  \sum_{u\in B}P_U(u)\sum_{y_2} P_{Y_2|U}(y_2|u)\left( \log\frac{P_{Y_2|U}(y_2|u)}{P_{Y_2}(y_2)} \right)^{\!2},
     \label{eq:i_UY2_conv4}\IEEEeqnarraynumspace
\end{IEEEeqnarray}
where we again partition $\mathcal{U}$ into the subsets $A$ and $B$ so that for  $u\in B$ we have $\limn P_U(u) = 0$ while for $u\in A$ we have $\limn P_U(u) >0$.

We now invoke again the arguments in the proof of Proposition~\ref{prop:decay_U}. The probability
of any $u'\in A$ is bounded from below, hence~(\ref{eq:uprime_condition2}), (\ref{eq:y2u_1}) and~(\ref{eq:y2u_2})
hold. Thus
\begin{equation}
\limn \left(\log\frac{P_{Y_2|U}(y_2|u)}{P_{Y_2}(y_2)}\right)^2 = (\log 1)^2 =0 \quad \forall u\in A,\label{eq:i_UY2_conv5}
\end{equation}
and the first sum in the r.h.s of~(\ref{eq:i_UY2_conv4}) vanishes as $n\rightarrow\infty$. Regarding the second
sum in the r.h.s of~(\ref{eq:i_UY2_conv4}), 
we use~(\ref{eq:y2u_2}) to write
\begin{IEEEeqnarray}{rCl}
\IEEEeqnarraymulticol{3}{l}{
\limn \sum_{u\in B}P_U(u)\sum_{y_2} P_{Y_2|U}(y_2|u)\left( \log\frac{P_{Y_2|U}(y_2|u)}{P_{Y_2}(y_2)} \right)^{\!2} }
\nonumber\\* \quad
    &=&
           \limn \sum_{u\in B}P_U(u)\sum_{y_2} P_{Y_2|U}(y_2|u)\left( \log\frac{P_{Y_2|U}(y_2|u)}{P_{Y_2|X}(y_2|x_0)} \right)^{\!2}. \nonumber\\
           \label{eq:i_UY2_conv6}
\end{IEEEeqnarray}
Since every $u$ is mapped randomly to  the elements of $\alphabetX$, we have by Part~c of Conditions~\ref{cond:Q0}
\begin{equation}
\frac{P_{Y_2|U}(y_2|u)}{P_{Y_2|X}(y_2|x_0)}  \leq \max_{y_2,x} \frac{P_{Y_2|X}(y_2|x)}{P_{Y_2|X}(y_2|x_0)} \leq M
\label{eq:i_UY2_conv7}
\end{equation}
for some finite, fixed $M$. Therefore by~(\ref{eq:i_UY2_conv6}), (\ref{eq:i_UY2_conv7}) 
\begin{IEEEeqnarray}{rCl}
\IEEEeqnarraymulticol{3}{l}{
\limn \sum_{u\in B}P_U(u)\sum_{y_2} P_{Y_2|U}(y_2|u)\left( \log\frac{P_{Y_2|U}(y_2|u)}{P_{Y_2}(y_2)} \right)^{\!2} }
\nonumber\\* \quad
    &\leq& \limn P_U(B) \log M = 0. \label{eq:i_UY2_conv8}
\end{IEEEeqnarray}
Using~(\ref{eq:i_UY2_conv5}) and (\ref{eq:i_UY2_conv8}) in~(\ref{eq:i_UY2_conv4}) yields~(\ref{eq:i_UY2_conv3}).
This establishes~(\ref{subeq:i_convergence_UY2}).

We proceed to prove~(\ref{subeq:i_convergence_XY1gU}). First, observe that following the lines of
the proof of~(\ref{subeq:i_convergence_UY2}), we also have
\begin{equation}
\sone i_{U^n;Y_1^n}(U^n;Y_1^n) \longrightarrow   \sqrt{n}I(U;Y_1)  \quad \mbox{in prob.},
                \label{subeq:i_convergence_UY1}\
\end{equation}
and, by previous results~\cite{WangWornellZheng:16p}
\begin{equation}
\sone i_{X^n;Y_1^n}(X^n;Y_1^n) \longrightarrow   \sqrt{n}I(X;Y_1)  \quad \mbox{in prob.}.
                \label{subeq:i_convergence_XY1_2}\
\end{equation}
Moreover
\begin{IEEEeqnarray}{rCl}
\IEEEeqnarraymulticol{3}{l}{
\sone i_{X^n;Y_1^n|U^n}(X^n;Y_1^n|U^n) } \nonumber\\* \quad
&=& \sone i_{X^n;Y_1^n}(X^n;Y_1^n) - \sone i_{U^n;Y_1^n}(U^n;Y_1^n)
\label{eq:i_XY1gU}
\end{IEEEeqnarray}
Then by properties of convergence in probability
\begin{IEEEeqnarray}{rCl}
\sone i_{X^n;Y_1^n|U^n}(X^n;Y_1^n|U^n) &\longrightarrow& \sqrt{n}I(X;Y_1) - \sqrt{n}I(U;Y_1)\nonumber\\
  &=& I(X;Y_1|U), \label{eq:i_convergence_XY1gU_2}
\end{IEEEeqnarray}
proving~(\ref{subeq:i_convergence_XY1gU}).

We have shown that for any $P_{UX}$ with marginal $P_{X}$ satisfying~(\ref{eq:limit_px0}), the random variables in~(\ref{eq:FeinsteinBC_n}) 
converge in probability
to the corresponding mutual information functions.
This implies that for any $\tgamma>0$ and any pairs~$\nu_1,\nu_2$ such that
\begin{subequations}
\label{subeq:f_L}
\begin{IEEEeqnarray}{rCl}
\frac{\log\nu_1}{\sqrt{n\delta}} &\leq& \sqrt{n/\delta} I(X;Y_1|U) - \tgamma/\delta \label{subeq:f_L1}\\
\frac{\log\nu_2}{\sqrt{n\delta}} &\leq& \sqrt{n/\delta} I(U;Y_2) - \tgamma/\delta \label{subeq:f_L2}\\
\frac{\log\nu_1\nu_2}{\sqrt{n\delta}} &\leq& \sqrt{n/\delta} I(X;Y_1) - \tgamma/\delta. 
\label{subeq:f_L1L2}
\end{IEEEeqnarray}
\end{subequations}
for some sequence $P_{U,X}$ satisfying the divergence constraint~(\ref{eq:div_constraint_limit}), there exists an~$(n,\nu_1,\nu_2,\epsilon,\delta)$ covert code for the BC.
The channel is degraded thus~(\ref{subeq:f_L1}),~(\ref{subeq:f_L2}) dominate~(\ref{subeq:f_L1L2}).
Since $\tgamma$ is arbitrary, this establishes the direct part.
\QEDB

\section{Proof of the Computable Region}
\label{sec:computable_proof}
We employ a Taylor expansions of $I(U;Y_2)$, $I(X;Y_1|U)$
and $D(P_Z||Q_0)$ near $\mu_1=0$, $\mu_2=0$. By~(\ref{subeq:structure_UX}), (\ref{subeq:structure_UX_2})
and~(\ref{subeq:structure_UX_3}) we have for $Y_k$, $k=1,2$ and $Z$:
\begin{subequations}
\label{subeq:joint_ditributions}
\begin{IEEEeqnarray}{rCl}
P_{X,Y_k}(x,y_k) &=& P_{Y_k|X}(y_k|x_0)\olsi{\mu}_1\olsi{\mu}_2\onei_{x_0}(x)  \nonumber \\
&& + P_{Y_k|X}(y_k|x)
\bigg[\olsi{\mu}_1\mu_2 \tilde{P}_X^A(x) + \mu_1 P_X^B(x)\bigg],\nonumber \\\label{subeq:joint_distributions_x_y}\\
P_{U,Y_k}(u,y_k) &=& \left\{ \,
                                       \begin{IEEEeqnarraybox}[] [c] {l?s}
                                       \IEEEstrut
                                       \olsi{\mu}_1\bigg[\olsi{\mu}_2 P_{Y_k|X}(y_k|x_0)  \nonumber \\ \hspace{.6cm} +
                                   \mu_2\tilde{P}_{Y_k|U}^A(y_k|u)\bigg]P_U^A(u) & for $u\in A$, \\
                                  \mu_1 P_{Y_k|U}^B(y_k|u) P_U^B(u) & for $u\in B.$
                                  \IEEEstrut
                                   \end{IEEEeqnarraybox}
                                   \right.\nonumber\\
                                   \label{subeq:joint_distributions_u_y} \\
 P_{Y_k}(y_k) &=& \olsi{\mu}_1\olsi{\mu}_2 P_{Y_k|X}(y_k|x_0) \nonumber\\
 && +\olsi{\mu}_1\mu_2\tilde{P}_{Y_k}^A(y_k)
                           + \mu_1 P_{Y_k}^B(y_k) \label{subeq:joint_distributions_yk}\\
P_Z(z) &=& \olsi{\mu}_1\olsi{\mu}_2 P_{Z|X}(z|x_0) \nonumber\\
&&+\sum_x P_{Z|X}(z|x)\left[\olsi{\mu}_1\mu_2 \tilde{P}_X^A(x)
                                    + \mu_1 P_X^B(x)\right]\nonumber\\
              &=& \olsi{\mu}_1\olsi{\mu}_2 P_{Z|X}(z|x_0)  +      \olsi{\mu}_1\mu_2 \tilde{P}_Z^A(z)
                                    + \mu_1 P_Z^B(z)               \nonumber\\
                                     \label{subeq:joint_distributions_z}                                   
\end{IEEEeqnarray}
where in~(\ref{subeq:joint_distributions_yk}) we used the definitions:
\begin{IEEEeqnarray}{rCl}
\tilde{P}_{Y_k}^A(y_k) &=& \sum_{u\in A} \tilde{P}_{Y_k|U}^A(y_k|u) P_U^A(u) \label{subeq:p_yk_A}\\
P_{Y_k}^B(y_k) &=& \sum_{u\in B} P_{Y_k|U}^B(y_k|u) P_U^B(u) \label{subeq:p_yk_B}
\end{IEEEeqnarray}
\end{subequations}
Note that
\begin{subequations}
\label{subeq:distribs_limit}
\begin{IEEEeqnarray}{rCl}
P_{X,Y_k}(x,y_k)\Bigr\rvert_{\substack{\mu_1 = 0\\ \mu_2=0}} &=& P_{Y_k,X}(y_k,x_0)\onei_{x_0}(x)
                                 \label{subeq:distribs_limit_XYk}\\
P_{U,Y_k}(u,y_k)\Bigr\rvert_{\substack{\mu_1 = 0\\ \mu_2=0}} &=& \left\{ \,
                                       \begin{IEEEeqnarraybox}[] [c] {l?s}
                                       \IEEEstrut
                                        P_{Y_k|X}(y_k|x_0) P_U^A(u)  & for $u\in A$, \\
                                  0 & for $u\in B.$
                                  \IEEEstrut
                                   \end{IEEEeqnarraybox}
                                   \right.\nonumber \\
                                   \label{subeq:distribs_limit_UYk} \\
P_{Y_k}(y_k)\Bigr\rvert_{\substack{\mu_1 = 0\\ \mu_2=0}}  &=& P_{Y_k|X}(y_k|x_0)\label{subeq:distribs_limit_Yk}\\
P_U(u)\Bigr\rvert_{\substack{\mu_1 = 0\\ \mu_2=0}}  &=& \left\{ \,
                                       \begin{IEEEeqnarraybox}[] [c] {l?s}
                                       \IEEEstrut
                                        P_U^A(u)  & for $u\in A$, \\
                                  0 & for $u\in B.$
                                  \IEEEstrut
                                   \end{IEEEeqnarraybox}
                                   \right.
                                   \label{subeq:distribs_limit_U}  \\
P_Z(z)\Bigr\rvert_{\substack{\mu_1 = 0\\ \mu_2=0}} &=& Q_0(z)\label{subeq:distribs_limit_Z}                                
\end{IEEEeqnarray}
\end{subequations}
We obtain the following derivatives of $I(U;Y_k)$:
\begin{subequations}
\label{subeq:d_I}
\begin{IEEEeqnarray}{rCl}
 \pdone{\mu_1} I(U;Y_k)\Bigr\rvert_{\substack{\mu_1 = 0\\ \mu_2=0}} 
 &=& \sum_{u\in B} P_U^B(u)D\left(P_{Y_k|U}^B(\cdot|u)||P_{Y_k|X}(\cdot|x_0)\right)\nonumber \\ \label{subeq:d_I_UY_mu1} \\
 \pdone{\mu_2} I(U;Y_k)\Bigr\rvert_{\substack{\mu_1 = 0\\ \mu_2=0}} &=& 0\label{subeq:d_I_UY_mu2}\\
  \pdone{\mu_1} I(X;Y_k)\Bigr\rvert_{\substack{\mu_1 = 0\\ \mu_2=0}} 
 &=& \sum_{x} P_X^B(x)D\left(P_{Y_k|X}(\cdot|x)||P_{Y_k|X}(\cdot|x_0)\right)\nonumber\\ \label{subeq:d_I_XY_mu1}\\
  \pdone{\mu_2} I(X;Y_k)\Bigr\rvert_{\substack{\mu_1 = 0\\ \mu_2=0}} 
 &=& \sum_{x} \tilde{P}_X^A(x)D\left(P_{Y_k|X}(\cdot|x)||P_{Y_k|X}(\cdot|x_0)\right).\nonumber\\ \label{subeq:d_I_XY_mu2}
\end{IEEEeqnarray}
\end{subequations}
Similarly, evaluating the derivatives of $D(P_Z||Q_0)$ w.r.t. $\mu_1$ and $\mu_2$, we obtain
\begin{equation}
\pdone{\mu_1} D(P_Z||Q_0)\Bigr\rvert_{\substack{\mu_1 = 0\\ \mu_2=0}} 
 = \pdone{\mu_2} D(P_Z||Q_0)\Bigr\rvert_{\substack{\mu_1 = 0\\ \mu_2=0}}  =0\label{eq:D_derivatives_1}
\end{equation}
and the Hessian
\begin{subequations}
\label{subeq:D_derivatives_2}
\begin{IEEEeqnarray}{rCl}
\pdtwo{\mu_1} D(P_Z||Q_0)\Bigr\rvert_{\substack{\mu_1 = 0\\ \mu_2=0}} 
                     &=&  \chi_2(P_Z^B||Q_0)\label{subeq:D_derivatives_2_mu1}\\
 \pdtwo{\mu_2} D(P_Z||Q_0)\Bigr\rvert_{\substack{\mu_1 = 0\\ \mu_2=0}} 
                     &=&  \chi_2(\tilde{P}_Z^A||Q_0)\label{subeq:D_derivatives_2_mu2}\\   
\pdcross{\mu_1}{\mu_2} D(P_Z||Q_0)\Bigr\rvert_{\substack{\mu_1 = 0\\ \mu_2=0}} 
                     &=&  \chi_2(\tilde{P}_Z^A,P_Z^B||Q_0),\label{subeq:D_derivatives_2_mu1_mu2}                                 
\end{IEEEeqnarray}
\end{subequations}
where 
\begin{equation} 
\chi_2(\tilde{P}_Z^A,P_Z^B||Q_0):=\sum_{z \in \mathcal{Z}} \frac{(\tilde{P}_Z^A(z) - Q_0(z))(P_Z^B(z)-Q_0(z))}{Q_0(z)}.
\end{equation} 

The derivations of~(\ref{subeq:d_I}), (\ref{eq:D_derivatives_1}) and~(\ref{subeq:D_derivatives_2}) are given
in Appendix~\ref{subappendix:proof_der_I_D_2}.
Thus we can write
\begin{subequations}
\label{subeq:computable_all}
\begin{IEEEeqnarray}{rCl}
I(X;Y_1|U) &=& I(X;Y_1) - I(U;Y_1)\nonumber\\
                  &=&  \mu_1 \sum_{x} P_X^B(x)D\left(P_{Y_1|X}^B(\cdot|x)||P_{Y_1|X}(\cdot|x_0)\right)\nonumber\\
                    & &    + \mu_2 \sum_{x} \tilde{P}_X^A(x)D\left(P_{Y_1|X}(\cdot|x)||P_{Y_1|X}(\cdot|x_0)\right)\nonumber\\
                    & & - \mu_1 \sum_{u\in B} P_U^B(u)D\left(P_{Y_1|U}^B(\cdot|u)||P_{Y_1|X}(\cdot|x_0)\right)\nonumber\\
                    & & + o(\mu_1,\mu_2)
                             \label{subeq:computable_all_1}\\
I(U;Y_2) &=& \mu_1 \sum_{u\in B} P_U^B(u)D\left(P_{Y_2|U}^B(\cdot|u)||P_{Y_2|X}(\cdot|x_0)\right)\nonumber\\
               & &  + o(\mu_1,\mu_2)
                             \label{subeq:computable_all_2}
\end{IEEEeqnarray}
\end{subequations}
where the divergence constraint (\ref{eq:div_constraint_limit}) entails
\begin{IEEEeqnarray}{rCl}
\lefteqn{\half\Big[\mu_1^2 \chi_2(P_Z^B||Q_0) + \mu_1\mu_2\chi_2(\tilde{P}_Z^A,P_Z^B||Q_0)} 
\nonumber \\
&& \hspace{2cm}  +\mu_2^2 \chi_2(\tilde{P}_Z^A||Q_0)\Big] + o(\mu_1^2,\mu_2^2) \leq \frac{\delta}{n} \label{eq:computable_all_3}
\end{IEEEeqnarray}
Using the normalization~(\ref{subeq:normalized_parameters}) 
in~(\ref{subeq:computable_all}) and~(\ref{eq:computable_all_3}) results in~(\ref{subeq:computable_region}).
\QEDB

\section{Inclusion of Time-Sharing Region ${\set L}^{(\textnormal{TS})} \subseteq {\set L}^{(I)}$ }\label{app:Ts}

Let $P_X^{1*}$ bet the $L_1^*$-achieving pmf and $P_X^{2*}$   the $L_2^*$-achieving pmf.  
(Both are pmfs over $\mathcal{X}\backslash \{x_0\}$). Let further $P_Z^{1*}$ and $P_Z^{2*}$   
be the corresponding output distributions at the warden. 

For any $\mu\in[0,1]$, 
specializing \eqref{subeq:computable_region_L1C}  and \eqref{subeq:computable_region_L2C}  to the choices $\tilde{P}_X^A=P_X^{1*}$ and $P_X^B=P_X^{2*}$ (so each $P_X^{\ell*}$ is only a pmf over $\mathcal{X}\backslash \{x_0\}$),   and choosing a deterministic mapping for $P_{U|X}^B$ results in the rate-pair 
\begin{IEEEeqnarray}{rCl}
L_1 &\leq&(1- \nu) \frac{\sqrt 2  \sum_x {P}_X^{1*}(x) D\left(P_{Y_1|X}(\cdot|x)||P_{Y_1|X}(\cdot|x_0) \right) }{  \sqrt{ \chi_{2} \left( (1-\nu) {P}_Z^{1*} +\nu {P}_Z^{2*} \big\| Q_0  \right) }}  \IEEEeqnarraynumspace
      \\ &  = & \alpha_1(\nu)L_1^*  \\[1.2ex]
L_2 &\leq& \nu\frac{ \sqrt{2}   \sum_x {P}_X^{2*}(x) D\left(P_{Y_2|X}(\cdot|x)||P_{Y_2|X}(\cdot|x_0) \right)}{\sqrt{ \chi_{2} \left( (1-\nu) {P}_Z^{1*} +\nu {P}_Z^{2*} \big\| Q_0  \right) }}    \\ & = &  {\alpha}_2(\nu)L_2^*
\end{IEEEeqnarray}
for 
\begin{IEEEeqnarray}{rCl}
 \alpha_1(\nu)& := & \frac{  (1-\nu) \sqrt{ \chi_{2} \left(  \tilde{P}_Z^{1*} \big\| Q_0  \right) }}{ \sqrt{ \chi_{2} \left( (1-\nu) \tilde{P}_Z^{1*} +\nu {P}_Z^{2*} \big\| Q_0  \right) }}\\
  \alpha_2(\nu)& := & \frac{  \nu \sqrt{ \chi_{2} \left(  \tilde{P}_Z^{2*} \big\| Q_0  \right) }}{ \sqrt{ \chi_{2} \left( (1-\nu) \tilde{P}_Z^{1*} +\nu {P}_Z^{2*} \big\| Q_0  \right) }}.
\end{IEEEeqnarray}
Varying $\nu$ from $0$ to $1$ varies $\alpha_1$ from $1$ to 0 and $\alpha_2$ from 0 to $1$. To show that this region includes the time-sharing region it suffices to show that for any $\nu \in[0,1]$:
\begin{equation} 
 \alpha_1(\nu) +  \alpha_2(\nu) \geq 1,
\end{equation} 
which holds because  $\alpha_1(\nu),   \alpha_2(\nu)  >0$ and because by the convexity of the square-root of the $\chi_{2}$-distance we have $(\alpha_1(\nu) +  \alpha_2(\nu))^2 \geq 1$, as proved by the sequence of Inequalities~\eqref{eq:alphas}--\eqref{eq:larger1} on top of the next page,
\begin{figure*}[h!]
\begin{IEEEeqnarray}{rCl}
(\alpha_1(\nu) +  \alpha_2(\nu))^2
& = &  \frac{  (1- \nu)^2 \chi_{2} \left(  \tilde{P}_Z^{1*} \big\| Q_0  \right) 
+ \nu^2  \chi_{2} \left(  \tilde{P}_Z^{2*} \big\| Q_0  \right)  + 2\nu (1-\nu) \sqrt{ \chi_{2} \left(  \tilde{P}_Z^{1*} \big\| Q_0  \right) }
     \sqrt{ \chi_{2} \left(  \tilde{P}_Z^{2*} \big\| Q_0  \right) } }{\chi_{2} \left( (1-\nu) \tilde{P}_Z^{1*} +\nu {P}_Z^{2*} \big\| Q_0  \right)} \label{eq:alphas}\\
&=& \frac{  (1- \nu)^2 \chi_{2} \left(  \tilde{P}_Z^{1*} \big\| Q_0  \right) + \nu^2  \chi_{2} \left(  \tilde{P}_Z^{2*} \big\| Q_0  \right)  
+ 2\nu (1-\nu) \sqrt{ \chi_{2} \left(  \tilde{P}_Z^{1*} \big\| Q_0  \right) }\sqrt{ \chi_{2} \left(  \tilde{P}_Z^{2*} \big\| Q_0  \right) } }{(1-\nu)^2\chi_{2} \left( \tilde{P}_Z^{1*}  
\big\| Q_0  \right) + \nu^2 \chi_{2} \left( \tilde{P}_Z^{2*}  \big\| Q_0  \right) 
+ 2 \nu (1-\nu) \sum_z \frac{(P_Z^{1*}(z)- Q_0(z))( P_Z^{2*}(z)- Q_0(z))   }{Q_0(z)}}\\
&\geq & 1 ,\label{eq:larger1}
\end{IEEEeqnarray}
\hrule 
\end{figure*}
where the last inequality holds because by Cauchy-Schwarz-Inequality: 
\begin{IEEEeqnarray}{rCl}
\lefteqn{\sqrt{ \sum_{z} \left( \frac{P_Z^{1*}(z)- Q_0(z)}{\sqrt{Q_0(z)}}  \right)^2 } 
\sqrt{ \sum_{z} \left( \frac{P_Z^{2*}(z)- Q_0(z)}{\sqrt{Q_0(z)}}  \right)^2 }}  \qquad  \nonumber\\
&
\geq& \sum_{z}  \frac{(P_Z^{1*}(z)- Q_0(z))}{\sqrt{Q_0(z)}}   \frac{(P_Z^{2*}(z)- Q_0(z))}{\sqrt{Q_0(z)}}. \hspace{1.6cm}
\end{IEEEeqnarray}

\section{Optimality of Time-sharing} \label{app:TS_opt}
For any set of achievable $(L_1, L_2)\in \tilde{\mathcal{L}}^{(I)}$, the  set of inequalities \eqref{eq:seq1}--\eqref{eq:last} on top of the next page holds for  some $B\subset\alphabetU$, $P_U^A$, $P_U^B$, $\tilde{P}_{X|U}$, $P_{X|U}$ and $\nu\in[0,1]$.
\begin{figure*}[h!]
\begin{IEEEeqnarray}{rCl}
\frac{L_1}{L_1^*} +\frac{ L_2}{L_2^*} 
  &\leq&\sqrt{\frac{2}{ \chi_{2}(\nu) }}   \left[ \nu  \sum_{u \in B}  P_U^B(u) \frac{I(Y_1; X^B| U=u) }{L_1^*}
         +(1-\nu) \sum_x \tilde{P}_X^A(x)\frac{ D\left(P_{Y_1|X}(\cdot|x)||P_{Y_1|X}(\cdot|x_0) \right) }{L_1^*}
        \right]  \nonumber \\
&& +  \sqrt{\frac{2}{ \chi_{2}(\nu) }}  \left[ \nu  \sum_x {P}_X^B(x) \frac{D\left(P_{Y_2|X}(\cdot|x)||P_{Y_2|X}(\cdot|x_0) \right) }{L_2^*} -    \nu \sum_{u\in B} P_U^B(u) \frac{I(Y_2;X^B|U=u) }{L_2^*} \right] \label{eq:seq1}\\
  &=&\sqrt{\frac{2}{ \chi_{2}(\nu) }}   \left[ \nu  \sum_{u \in B}  P_U^B(u) \frac{I(Y_2; X^B| U=u) }{L_2^*}
         +(1-\nu) \sum_x \tilde{P}_X^A(x)\frac{ D\left(P_{Y_2|X}(\cdot|x)||P_{Y_2|X}(\cdot|x_0) \right) }{L_2^*}
        \right]  \nonumber \\
&& +  \sqrt{\frac{2}{ \chi_{2}(\nu) }}  \left[ \nu  \sum_x {P}_X^B(x) \frac{D\left(P_{Y_2|X}(\cdot|x)||P_{Y_2|X}(\cdot|x_0) \right) }{L_2^*} -    \nu \sum_{u\in B} P_U^B(u) \frac{I(Y_2;X^B|U=u) }{L_2^*} \right]\\
          & = & (L_2^*)^{-1}  \sqrt{\frac{2}{ \chi_{2}(\nu) }} \bigg( \nu    \sum_x {P}_X^B(x) D\left(P_{Y_2|X}(\cdot|x)||P_{Y_2|X}(\cdot|x_0) \right)  \nonumber \\
                  && \hspace{4cm}+ (1-\nu)  \sum_x \tilde{P}_X^A(x) D\left(P_{Y_2|X}(\cdot|x)||P_{Y_2|X}(\cdot|x_0) \right) \bigg) \label{eq:second_last} \\
                  & \leq & 1,\label{eq:last}
\end{IEEEeqnarray}
\hrule
\end{figure*}
where in the  first  equality we applied \eqref{eq:ineq},  and in the second inequality we used the fact that the rate in \eqref{eq:second_last} corresponds to the  rate to User 2 achieved by a  time-sharing scheme employing  pmf $\tilde{P}_X^A$ during $(1-\nu)$-fraction of the time and pmf ${P}_X^B$ during the remaining time, which cannot exceed $L_2^*$.

This establishes optimality of  time-sharing as proved in \cite{TanLee:19p}.

\section{Derivation of Mutual Informations and Divergence Derivatives} 
\subsection{Proof of~(\ref{eq:der_1_I})}
\label{appendix:proof_der1_I}
We have to show  that
\begin{equation}
\sum_{u,y_2} P_{U,Y_2}(u,y_2)\pdone{\mu} \log\frac{P_{U,Y_2}(u,y_2)}{P_U(u)P_{Y_2}(y_2)} = 0.\label{eq:der_1_I_1}
\end{equation}
Let $P_{U|X}$ be a general conditional distribution. By~(\ref{eq:Phat_X})
\begin{subequations}
\label{subeq:general_distributions}
\begin{IEEEeqnarray}{rCl}
P_{U,X}(u,x) &=& (1-\mu)P_{U|X}(u|x_0) \onei_{x_0}(x)  \nonumber\\
&& +\mu P_{U|X}(u|x)\tilde{P}_X(x)  \nonumber\\\label{subeq:pux_general}\\
P_{U,Y_2}(u,y_2) &=& (1-\mu)P_{U|X}(u|x_0)P_{Y_2|X}(y_2|x_0) \nonumber\\
&&+ \mu\sum_xP_{U|X}(u|x)\tilde{P}_X(x)P_{Y_2|X}(y_2|x) \IEEEeqnarraynumspace
   \label{subeq:puy_general}\\
   P_U(u) &=& (1-\mu) P_{U|X}(u|x_0) +\mu\sum_xP_{U|X}(u|x)\tilde{P}_X(x) \label{subeq:pu_general} \nonumber\\\\
   P_{Y_2}(y_2) &=& (1-\mu) P_{Y_2|X}(y|x_0)  \nonumber\\
   &&+\mu\sum_x\tilde{P}_X(x)P_{Y_2|X}(y|x)\label{subeq:py_general}
\end{IEEEeqnarray}
\end{subequations}
Write
\begin{IEEEeqnarray}{rCl}
\IEEEeqnarraymulticol{3}{l}{
\sum_{u,y_2} P_{U,Y_2}(u,y_2)\pdone{\mu} \log\frac{P_{U,Y_2}(u,y_2)}{P_U(u)P_{Y_2}(y_2)}}  \nonumber\\* \quad 
&= &  \sum_{u,y_2} P_{U,Y_2}(u,y_2)\left[ \frac{\pdone{\mu} P_{U,Y_2}(u,y_2)}{P_{U,Y_2}(u,y_2) } 
 - \frac{\pdone{\mu} P_{U}(u)}{P_{U}(u) } \right.\nonumber\\
 && \left. \hspace{4.8cm} - \frac{\pdone{\mu} P_{Y_2}(y_2)}{P_{Y_2}(y_2) } \right]\IEEEeqnarraynumspace\\
 &=& \pdone{\mu}\sum_{u,y_2} P_{U,Y_2}(u,y_2) - \sum_{u,y_2} P_{Y_2|U}(y_2|u) \pdone{\mu}P_U(u)
        \nonumber\\
 &&    - \sum_{u,y_2} P_{U|Y_2}(u|y_2) \pdone{\mu}P_{Y_2}(y_2)
\label{eq:der_1_I_2}
\end{IEEEeqnarray}
Using~(\ref{subeq:pu_general}) and~(\ref{subeq:py_general}) we obtain
\begin{subequations}
\label{subeq:der_1_I}
\begin{IEEEeqnarray}{rCl}
\IEEEeqnarraymulticol{3}{l}{
\sum_{u,y_2} P_{Y_2|U}(y_2|u) \pdone{\mu}P_U(u)} \nonumber\\* 
                         &=& \sum_{u,y_2} P_{Y_2|U}(y_2|u)\left[-P_{U|X}(u|x_0)
                            + \sum_xP_{U|X}(u|x)\tilde{P}_X(x) \right]\nonumber\\
                            & =&  1-1 = 0 \label{eq:der_1_I_3}
\end{IEEEeqnarray}
and
\begin{IEEEeqnarray}{rCl}
\IEEEeqnarraymulticol{3}{l}{
\sum_{u,y_2} P_{U|Y_2}(u|y_2) \pdone{\mu}P_{Y_2}(y_2) }\nonumber\\* 
                         &=& \sum_{u,y_2} P_{U|Y_2} (u|y_2) \bigg[ -P_{Y_2|X}(y_2|x_0)   \nonumber\\
                         && + \sum_x \tilde{P}_X(x) P_{Y_2|X}(y_2|x)\bigg] = 0.
                          \label{eq:der_1_I_4} 
\end{IEEEeqnarray}
\end{subequations}
Substitution of~(\ref{subeq:der_1_I}) in~(\ref{eq:der_1_I_2}) yields the desired result.
\QEDB

\subsection{Proofs of \eqref{subeq:d_I}--\eqref{subeq:D_derivatives_2}}
\label{subappendix:proof_der_I_D_2}
\subsubsection{Proof of~(\ref{subeq:d_I})}
\label{subappendix:proof_der_I}
We first present general derivative formulas for the mutual information functions.
\begin{IEEEeqnarray}{rCl}
\lefteqn{\pdone{\mu_j} I(U;Y_k) } \nonumber \\&=& \sum_{u,y_k}\left[\pdone{\mu_j}P_{U,Y_k}(u,y_k)\right]
                  \log\frac{P_{U,Y_k}(u,y_k)}{P_U(u)P_{Y_k}(y_k)}\nonumber\\
                  & & +\sum_{u,y_k} P_{U,Y_k}(u,y_k)\bigg[\frac{\pdone{\mu_j}P_{U,Y_k}(u,y_k)}{P_{U,Y_k}(u,y_k)}
                       \nonumber \\
                       && \hspace{2cm}   -\frac{\pdone{\mu_j}P_U(u)}{P_U(u)} -\frac{\pdone{\mu_j}P_{Y_k}(y_k)}{P_{Y_k}(y_k)}\bigg]
                        \\
                   &=& A_{1,k,j} +A_{2,k,j},\quad k=1,2,\quad j=1,2.\label{eq:appB_1}
\end{IEEEeqnarray}
with the obvious definitions for $A_{1,j}$ and $A_{2,j}$. Evaluating these terms:
\begin{IEEEeqnarray}{rCl}
\lefteqn{A_{2,k,j} }\nonumber \\&=& -\sum_{u,y_k} P_{U,Y_k}(u,y_k)\bigg[\frac{\pdone{\mu_j}P_{U,Y_k}(u,y_k)}{P_{U,Y_k}(u,y_k)}
                     \nonumber \\
                     && \hspace{3.4cm}   -\frac{\pdone{\mu_j}P_U(u)}{P_U(u)} -\frac{\pdone{\mu_j}P_{Y_k}(y_k)}{P_{Y_k}(y_k)}\bigg]
                        \\
         &=& -\sum_{u,y_k}\bigg[ P_{Y_k|U}(y_k|u)\pdone{\mu_j}P_U(u)  \nonumber \\
         && \hspace{1cm} + 
         P_{U|Y_k}(u|y_k)\pdone{\mu_j} P_{Y_k}(y_k)\bigg],\quad k=1,2,\quad j=1,2.\nonumber \\\label{eq:appB_2}
\end{IEEEeqnarray}
Evaluating the derivatives in the r.h.s. of~(\ref{eq:appB_2}), we have
\begin{IEEEeqnarray}{rCl}
\pdone{\mu_1} P_U(u) &=& -P_U^A(u) + P_U^B(u)\label{eq:appB_3}\\
\pdone{\mu_1} P_{Y_k}(y_k) &=& -\olsi{\mu}_2 P_{Y_k|X}(y_k|x_0) - \mu_2 \tilde{P}_{Y_k}^A(y_k) + P_{Y_k}^B(y_k) \nonumber \\
   \label{eq:appB_4}\\
   \pdone{\mu_2} P_U(u) &=& 0\label{eq:appB_3_2}\\
   \pdone{\mu_2} P_{Y_k}(y_k) &=& \olsi{\mu}_1\left[ \tilde{P}_{Y_k}^A(y_k) - P_{Y_k|X}(y_k|x_0)\right] \label{eq:appB_4_2}
\end{IEEEeqnarray}
resulting in 
\begin{equation}
A_{2,k,j}=0.  \label{eq:appB_4_3}
\end{equation}
For $A_{1,k,j}$, we first evaluate the derivative of the joint distribution:
\begin{IEEEeqnarray}{rCl}
\lefteqn{
\pdone{\mu_1}P_{U,Y_k}(u,y_k) } \qquad \nonumber \\
&=& \left\{ \,
                                       \begin{IEEEeqnarraybox}[] [c] {l?s}
                                       \IEEEstrut
                                       -\Big[\olsi{\mu}_2 P_{Y_k|X}(y_k|x_0)  \nonumber \\
                                       \hspace{1.2cm}+
                                   \mu_2\tilde{P}_{Y_k|U}^A(y_k|u)\Big] P_U^A(u) & for $u\in A$, \\
                                  P_{Y_k|U}^B(y_k|u) P_U^B(u) & for $u\in B.$
                                  \IEEEstrut
                                   \end{IEEEeqnarraybox}
                                   \right.
                                   \label{eq:appB_5}
\end{IEEEeqnarray}
\begin{IEEEeqnarray}{rCl}
\lefteqn{\pdone{\mu_2}P_{U,Y_k}(u,y_k) } \nonumber \\
&=& \left\{ \,
                                       \begin{IEEEeqnarraybox}[] [c] {l?s}
                                       \IEEEstrut
                                       \olsi{\mu}_1\left[- P_{Y_k|X}(y_k|x_0) +
                                   \tilde{P}_{Y_k|U}^A(y_k|u)\right] P_U^A(u) & for $u\in A$, \\
                                   0 & for $u\in B.$
                                  \IEEEstrut
                                   \end{IEEEeqnarraybox}
                                   \right.\nonumber \\
                                   \label{eq:appB_5_2}
\end{IEEEeqnarray}                                   
Therefore
\begin{IEEEeqnarray}{rCl}
\lefteqn{A_{1,k,1}\Bigr\rvert_{\substack{\mu_1 = 0\\ \mu_2=0}} } \qquad  \nonumber \\&=& \sum_{y_k,u\in A} \left[
   -\olsi{\mu}_2 P_{Y_k|X}(y_k|x_0) -
                                   \mu_2\tilde{P}_{Y_k|U}^A(y_k|u)\right]  \nonumber \\
                                   & & \hspace{1cm} \cdot P_U^A(u)  
                                   \log\frac{P_{U,Y_k}(u,y_k)}{P_U(u) P_{Y_k}(y_k)}\Bigr\rvert_{\substack{\mu_1 = 0\\ \mu_2=0}} 
                                   \nonumber\\
                                   & & +\sum_{y_k,u\in B} P_{Y_k|U}^B(y_k|u) P_U^B(u)
                                   \log\frac{ P_{Y_k|U}^B(y_k|u)}{P_{Y_k|X}(y_k|x_0)}\IEEEeqnarraynumspace\\
                                   &=& \sum_{u\in B} P_U^B(u) D\left(P_{Y_k|U}^B(\cdot|u)||P_{Y_k|X}(\cdot|x_0)\right)
                                   \label{eq:appB_6}
\end{IEEEeqnarray}
where in the last equality we used~(\ref{subeq:distribs_limit}). Subsituting~(\ref{eq:appB_6})
and~(\ref{eq:appB_4_3}) in~(\ref{eq:appB_1}) proves~(\ref{subeq:d_I_UY_mu1}).

For~(\ref{subeq:d_I_UY_mu2}), we only have to evaluate $A_{1,k,2}$. By~(\ref{eq:appB_2})   we have
\begin{IEEEeqnarray}{rCl}
\lefteqn{A_{1,k,2}\Bigr\rvert_{\substack{\mu_1 = 0\\ \mu_2=0}} } \nonumber \\
&=& \sum_{u,y_k}\left[\pdone{\mu_2}P_{U,Y_k}(u,y_k)\right]
                  \log\frac{P_{U,Y_k}(u,y_k)}{P_U(u)P_{Y_k}(y_k)}\Bigr\rvert_{\substack{\mu_1 = 0\\ \mu_2=0}} \IEEEeqnarraynumspace\\
                  &=& 0\label{eq:appB_10}
\end{IEEEeqnarray}
where we used~(\ref{subeq:distribs_limit}) and the fact that the r.h.s of~(\ref{eq:appB_5_2}) is bounded.
Subsituting~(\ref{eq:appB_10})
and~(\ref{eq:appB_4_3}) in~(\ref{eq:appB_1}) proves~(\ref{subeq:d_I_UY_mu2}).

The proof of~(\ref{subeq:d_I_XY_mu1}) and~(\ref{subeq:d_I_XY_mu2}) proceed along the same lines
as that of~(\ref{subeq:d_I_UY_mu1}). Parallel to~(\ref{eq:appB_1}), we write
\begin{IEEEeqnarray}{rCl}
\lefteqn{\pdone{\mu_j} I(X;Y_k) } \nonumber \quad \\
&=& \sum_{x,y_k}\left[\pdone{\mu_j}P_{X,Y_k}(x,y_k)\right]
                  \log\frac{P_{X,Y_k}(x,y_k)}{P_X(x)P_{Y_k}(y_k)}\nonumber\\
                  & & +\sum_{x,y_k} P_{X,Y_k}(x,y_k)\Bigg[\frac{\pdone{\mu_j}P_{X,Y_k}(x,y_k)}{P_{X,Y_k}(x,y_k)}
                   \nonumber \\
                   && \hspace{3.3cm}     -\frac{\pdone{\mu_j}P_X(x)}{P_X(x)} -\frac{\pdone{\mu_j}P_{Y_k}(y_k)}{P_{Y_k}(y_k)}\Bigg]
                        \nonumber\\
                   &=& B_{1,k,j} +B_{2,k,j},\hspace{1.3cm} k=1,2,\quad j=1,2.\label{eq:appB_11}
\end{IEEEeqnarray}
and
\begin{IEEEeqnarray}{rCl}
B_{2,k,j} &=& -\sum_{x,y_k}\bigg[ P_{Y_k|X}(y_k|x)\pdone{\mu_j}P_X(x)  \nonumber \\
&& \hspace{2cm}+ 
         P_{X|Y_k}(x|y_k)\pdone{\mu_j} P_{Y_k}(y_k)\bigg],\nonumber \\
         && \hspace{3.3cm} k=1,2,\quad j=1,2. \IEEEeqnarraynumspace\label{eq:appB_12} 
\end{IEEEeqnarray}
Evaluation of the derivatives in the r.h.s. of~(\ref{eq:appB_12}) gives
\begin{IEEEeqnarray}{rCl}
\pdone{\mu_1} P_X(x) &=& -\left[\olsi{\mu}_2\onei_{x_0}(x) +\mu_2 \tilde{P}_X^A(x)\right] + P_X^B(x)
                                               \label{eq:appB_13}\\
 \pdone{\mu_2} P_X(x) &=& \olsi{\mu}_1\left[-\onei_{x_0}(x) +\tilde{P}_X^A(x) \right].
                                                \label{eq:appB_14} 
\end{IEEEeqnarray}
Using~(\ref{eq:appB_13}) and~(\ref{eq:appB_4}) in~(\ref{eq:appB_12})
yields
\begin{IEEEeqnarray}{rCl}
B_{2,k,1} &=& \sum_{x,y_k} \left[ P_{Y_k|X}(y_k|x) \left( -\olsi{\mu}_2\onei_{x_0}(x) -\mu_2\tilde{P}_X^A(x) 
                                                     +P_X^B(x)\right) \right. \nonumber\\
                        & &\hspace{1cm}  +P_{X|Y_k}(x|y_k) \bigg(
                         -\olsi{\mu}_2 P_{Y_k|X}(y_k|x_0) \nonumber \\
&& \hspace{3cm}  \left. - \mu_2 \tilde{P}_{Y_k}^A(y_k) + P_{Y_k}^B(y_k)
                         \bigg)
                         \right]\nonumber\\
                         &=& 0.\label{eq:appB_15}
\end{IEEEeqnarray}
Similarly, using~(\ref{eq:appB_14}) and~(\ref{eq:appB_4_2}) in~(\ref{eq:appB_12}) yields
\begin{IEEEeqnarray}{rCl}
B_{2,k,2} &=& \sum_{x,y_k} \left[ P_{Y_k|X}(y_k|x) \left( -\olsi{\mu}_1\onei_{x_0}(x) +\olsi{\mu}_1\tilde{P}_X^A(x) 
                                                    \right) \right. \nonumber\\
                        & & + \left. P_{X|Y_k}(x|y_k) \left(
                         \olsi{\mu}_1 \tilde{P}_{Y_k}^A(y_k) - \olsi{\mu}_1 P_{Y_k|X}(y_k|x_0) 
                         \right)
                         \right]\nonumber\\
                         &=& 0.\label{eq:appB_16}
\end{IEEEeqnarray}
Next we evaluate $B_{1,k,j}$. By~(\ref{subeq:joint_distributions_x_y})
\begin{IEEEeqnarray}{rCl}
\pdone{\mu_1} P_{X,Y_k}(x,y_k) &=& \left[-\olsi{\mu}_2\onei_{x_0}(x) -\mu_2\tilde{P}_X^A(x) +P_X^B(x)\right] \nonumber \\
&& \cdot P_{Y_k|X}(y_k|x)
                                                             \label{eq:appB_17}\\
\pdone{\mu_2} P_{X,Y_k}(x,y_k) &=&  \olsi{\mu}_1\left[-\onei_{x_0}(x) +\tilde{P}_X^A(x)\right] P_{Y_k|X}(y_k|x), \nonumber\\
                                                             \label{eq:appB_18}                                                         
\end{IEEEeqnarray}
hence
\begin{IEEEeqnarray}{rCl}
B_{1,k,1} &=& \sum_{x,y_k}  \left[-\olsi{\mu}_2\onei_{x_0}(x) -\mu_2\tilde{P}_X^A(x) +P_X^B(x)\right] \nonumber \\
&& \cdot P_{Y_k|X}(y_k|x)
                          \log\frac{P_{Y_k|X}(y_k|x)}{P_{Y_k}(y_k)} \label{eq:appB_19}\\
\end{IEEEeqnarray}
Using~(\ref{subeq:distribs_limit_Yk}) we arrive at
\begin{equation}
B_{1,k,1}\Bigr\rvert_{\substack{\mu_1 = 0\\ \mu_2=0}} = \sum_x P_X^B(x) D\left( P_{Y_k|X}(\cdot|x)|| P_{Y_k|X}(\cdot|x_0)\right).
                                                 \label{eq:appB_20}
\end{equation}
Now~(\ref{subeq:d_I_XY_mu1}) follows from~(\ref{eq:appB_11}), (\ref{eq:appB_15}) and~(\ref{eq:appB_20}). 
Similarly, by~(\ref{eq:appB_18})
\begin{IEEEeqnarray}{rCl}
B_{1,k,2} &=& \olsi{\mu}_1\sum_{x,y_k} \left[-\onei_{x_0}(x) +\tilde{P}_X^A(x)\right] \nonumber \\
&& \hspace{1.2cm} \cdot  P_{Y_k|X}(y_k|x)
                       \log\frac{P_{Y_k|X}(y_k|x)}{P_{Y_k}(y_k)} \label{eq:appB_21}  
\end{IEEEeqnarray}
which, using again~(\ref{subeq:distribs_limit_Yk}), implies
\begin{equation}
B_{1,k,2}\Bigr\rvert_{\substack{\mu_1 = 0\\ \mu_2=0}} = \sum_x \tilde{P}_X^A(x) D\left( P_{Y_k|X}(\cdot|x)|| P_{Y_k|X}(\cdot|x_0)\right).
                                                 \label{eq:appB_22}
\end{equation}
The proof of~(\ref{subeq:d_I_XY_mu2}) follows from~(\ref{eq:appB_11}), (\ref{eq:appB_16}) and~(\ref{eq:appB_22}).

\subsubsection{Proof of~(\ref{eq:D_derivatives_1})}
\label{subappendix:proof_D_der_1}
\begin{equation}
\pdone{\mu_j} D(P_Z||Q_0) = \sum_z \left(\pdone{\mu_j}P_Z(z)\right)\log\frac{P_Z(z)}{Q_0(z)} \label{eq:appB_23}
\end{equation}
where, by~(\ref{subeq:joint_distributions_z})
\begin{IEEEeqnarray}{rCl}
\pdone{\mu_1} P_Z(z) &=& -\olsi{\mu}_2 Q_0(z) -\mu_2\tilde{P}_Z^A(z) + P_Z^B(z)\label{eq:appB_24}\\
\pdone{\mu_2} P_Z(z) &=&  -\olsi{\mu}_1 Q_0(z) +\olsi{\mu}_1 \tilde{P}_Z^A(z). \label{eq:appB_25}
\end{IEEEeqnarray}
Note that the derivatives~(\ref{eq:appB_24}), (\ref{eq:appB_25}) are bounded. Therefore~(\ref{eq:D_derivatives_1})
follows from~(\ref{eq:appB_23}) and~(\ref{subeq:distribs_limit_Z}).

\subsubsection{Proof of~(\ref{subeq:D_derivatives_2})}
\label{subappendix:proof_D_der_2}
\begin{IEEEeqnarray}{rCl}
\lefteqn{\pdtwo{\mu_j} D(P_Z||Q_0) } \nonumber \\ & =& \sum_z \left(\pdtwo{\mu_j} P_Z(z) \right) \log\frac{P_Z(z)}{Q_0(z)} 
                                                  +  \sum_z \left(\pdone{\mu_j} P_Z(z) \right)^2 \frac{1}{P_Z(z)}    \nonumber \\
                                                  \label{eq:appB_26}
\end{IEEEeqnarray}
Since the second derivatives of $P_Z(z)$ according to $\mu_j$ are bounded, (\ref{subeq:distribs_limit_Z})
implies that the first sum in the r.h.s. of~(\ref{eq:appB_26}) is 0. Using~(\ref{eq:appB_24}), (\ref{eq:appB_25}),
in the second sum of~(\ref{eq:appB_26}) we get
\begin{IEEEeqnarray}{rCl}
\pdtwo{\mu_1} D(P_Z||Q_0)\Bigr\rvert_{\substack{\mu_1 = 0\\ \mu_2=0}} &=& \sum_z\frac{(P_Z^B(z)- Q_0(z))^2}{Q_0(z)}
                                                     \\
                                                     & =&\chi_2(P_Z^B||Q_0)\label{eq:appB_27}\\
\pdtwo{\mu_2} D(P_Z||Q_0)\Bigr\rvert_{\substack{\mu_1 = 0\\ \mu_2=0}} &=& \sum_z\frac{(\tilde{P}_Z^A(z)- Q_0(z))^2}{Q_0(z)}
                                                    \\& =&\chi_2(\tilde{P}_Z^A||Q_0)\label{eq:appB_28}
\end{IEEEeqnarray}
proving~(\ref{subeq:D_derivatives_2_mu1}) and~(\ref{subeq:D_derivatives_2_mu2}). For~(\ref{subeq:D_derivatives_2_mu1_mu2})
\begin{IEEEeqnarray}{rCl}
\pdcross{\mu_1}{\mu_2} D(P_Z||Q_0) &=& \sum_z\left(\pdcross{\mu_1}{\mu_2}P_Z(z)\right)\log\frac{P_Z(z)}{Q_0(z)} \nonumber\\
                & & + \sum_z \left(\pdone{\mu_1} P_Z(z)\right)\frac{1}{P_Z(z)}\pdone{\mu_2} P_Z(z)\nonumber\\\\
                &=& C_1 +C_2
                 \label{eq:appB_29}
\end{IEEEeqnarray}
where $C_1$ (resp. $C_2$) is the first (resp. second) sum in~(\ref{eq:appB_29}). Then, using~(\ref{subeq:distribs_limit_Z})
and the derivatives~(\ref{eq:appB_24}), (\ref{eq:appB_25})
we obtain
\begin{IEEEeqnarray}{rCl}
C_1\Bigr\rvert_{\substack{\mu_1 = 0\\ \mu_2=0}} &=& 0,\label{eq:appB_30}\\
C_2\Bigr\rvert_{\substack{\mu_1 = 0\\ \mu_2=0}} &=& \sum_z\frac{(P_Z^B(z) - Q_0(z))
                                                                                      (\tilde{P}^A_Z(z) - Q_0(z))}{Q_0(z)}\IEEEeqnarraynumspace
                                                                                  \\
                                                                                  &    =&\chi_2(\tilde{P}_Z^A, P_Z^B||Q_0) \label{eq:appB_31}
\end{IEEEeqnarray}
The proof of~(\ref{subeq:D_derivatives_2_mu1_mu2}) follows from~(\ref{eq:appB_29}), (\ref{eq:appB_30})
and (\ref{eq:appB_31}).
\QEDB

\end{appendices}

\bibliographystyle{hieeetr}
\bibliography{./references_Yossi_2025_09_16.bib}

\end{document}